\documentclass[conference]{IEEEtran}
\IEEEoverridecommandlockouts
\usepackage{cite}
\usepackage{amsmath,amssymb,amsfonts,amsthm}
\usepackage{bbm}
\usepackage{mathtools}
\usepackage{physics}
\usepackage{algorithmic}
\usepackage{enumitem}
\usepackage{empheq}
\usepackage{multirow}
\usepackage{hyperref}
\usepackage[skip=2.5pt]{caption}
\usepackage{graphicx}
\usepackage{subcaption}
\usepackage{textcomp}
\usepackage{nicefrac}
\usepackage{accents}
\usepackage[dvipsnames]{xcolor}

\DeclarePairedDelimiter\ceil{\lceil}{\rceil}

\newtheorem{proposition}{Proposition}
\newtheorem{definition}{Definition}[section]
\newtheorem{remark}{Remark}

\def\BibTeX{{\rm B\kern-.05em{\sc i\kern-.025em b}\kern-.08em
    T\kern-.1667em\lower.7ex\hbox{E}\kern-.125emX}}
\begin{document}
\title{On Utility-optimal Entanglement Routing in Quantum Networks}
\author{
\IEEEauthorblockN{Sounak Kar}
\IEEEauthorblockA{QuTech, TU Delft}
\and
\IEEEauthorblockN{Arpan Mukhopadhyay}
\IEEEauthorblockA{University of Warwick}
}

\maketitle
\looseness = -1
\begin{abstract}
Quantum networks are envisioned to enable reliable distribution and manipulation of quantum information across distances, forming the foundation of a future quantum internet.
The fair and efficient allocation of communication resources in such networks has been addressed through the quantum network utility maximization (QNUM) framework, which optimizes network utility under the assumption of predetermined routes for competing user demands.
In this work, we relax this assumption and aim to identify optimal routes that correspond to the maximum achievable network utility. 
Specifically, we formulate the single-path utility-based entanglement routing problem as a Mixed-Integer Convex Program (MICP).
The formulation is exact when negativity is chosen as the entanglement measure for utility quantification \textit{or} the network supports sufficiently high entanglement generation rates across demands.
For other entanglement measures considered, the formulation approximates the problem with over $99.99\%$ accuracy on evaluated real-world examples.
To improve computational tractability, we propose a randomized rounding-based heuristic and an upper bound via the relaxation of the MICP.
Furthermore, based on min-congestion routing, we introduce an alternative randomized heuristic and upper bound.
This heuristic is computationally faster, while both the heuristic and the upper bound often outperform their counterparts on considered real-world networks.
Our work provides the framework for extending classical flow-based and quality of service-aware routing concepts to quantum networks.
\end{abstract}


\section{Introduction}\label{sec:intro}
% \vspace{-2pt}
\everypar{\looseness=-1}
Quantum networks hold the promise of facilitating tasks~\cite{bennett1984quantum,giovannetti2004quantum,broadbent2009universal} that are fundamentally beyond the reach of classical networks.
The performance of a communication network at both user and network levels is typically characterized by Quality of Service (QoS) metrics that capture the needs of underlying applications.
A common framework for optimizing QoS at the network level is network utility maximization (NUM)\cite{kelly1997charging,kelly1998rate}.
NUM quantifies user satisfaction through a utility function that maps the allocated network resources to individual utilities, which are then aggregated at the network level.
The network utility is subsequently maximized over all feasible resource allocations to achieve the most \textit{fair} and \textit{efficient} resource distribution.
In contrast to classical networks, where transmission rate typically constitutes the primary communication resource, user utility in quantum networks fundamentally depends on both entanglement generation rate and fidelity.
This problem is termed quantum network utility maximization (QNUM)\cite{vardoyan2023quantum}, and approaches for its efficient solution have been investigated in~\cite{kar2024convexification}.

\looseness = -1 Notably, QNUM addresses network utility maximization under the assumption that routes are predetermined for each demand, represented by pairs of end nodes.
In this work, we relax this assumption and seek to identify routes that correspond to the maximum achievable network utility.
In other words, our objective is to \textit{jointly} optimize network utility over both feasible routes and resource allocations.
Such problems have been studied in classical networks, arising naturally in the contexts of QoS routing~\cite{apostolopoulos1998quality}, optimal network pricing~\cite{lin2003multi}, and opportunistic routing in wireless networks~\cite{wu2008utility}.
Utility-based routing can be viewed as a generalization of network flow models~\cite{ahuja1993network}, which have been extensively employed for diverse objectives, including throughput maximization and congestion minimization, and remain relevant in modern networks~\cite{charikar2018multi}.
A general framework for addressing such problems was established through multipath utility maximization~\cite{lin2006utility}, which requires the candidate routes for each demand to be precomputed.
However, in general, the number of feasible routes for a given demand grows exponentially with network size.
% It is noteworthy that QNUM considers maximising network utility when routes are given for each demand, represented by pairs of end nodes.
% In this work, we take away the assumption of fixed routes and aim to find routes for which the highest achievable network utility is maximised.
% In other words, our objective is to maximise the network utility over possible routes and allocations.
% Such problems were studied in classical networks, which naturally arise in the context of QoS-routing~\cite{apostolopoulos1998quality}, optimal network pricing~\cite{lin2003multi} and opportunistic routing in wireless networks~\cite{wu2008utility}.
% Utility-based routing can be seen as a generalisation of network flow models of routing~\cite{ahuja1993network}, which has been extensively used for varying objectives such as throughput maximisation and congestion minimisation and continue to find applications 
% even in modern networks~\cite{charikar2018multi}.
% A general framework for handling such problems was laid out via multipath utility maximisation~\cite{lin2006utility}.
% Notably,~\cite{lin2006utility} requires the paths corresponding to each demand to be precomputed, and, in general, the number of paths for a given demand is exponential in the size of the network.

In the context of quantum networks, the problem of routing entanglement has recently received significant attention.
Rather than providing a comprehensive review of this expanding literature (see~\cite{abane2025entanglement} and references therein for an exhaustive overview), we focus on summarizing key approaches and clarifying how their objectives differ from ours.
Ref.~\cite{van2013path} proposes a method to minimize the number of quantum measurements required to establish an end-to-end (e2e) link by modifying cost inputs to Dijkstra's algorithm suitably, while~\cite{gyongyosi2017entanglement} proposes a decentralized method for shortest-path routing.
A majority of the literature has focused on optimization of entanglement generation rate.
For instance,~\cite{pant2019routing} presents greedy strategies for maximizing generation rate on 2D grids and extending to multi-demand rate region maximization, whereas~\cite{chakraborty2020entanglement} applies a multi-commodity flow formulation to maximize generation rate under multiple demands, enforcing fidelity guarantees by limiting hop counts. 
On the fidelity side,~\cite{li2022fidelity} derives low-complexity purification-enabled algorithms providing fidelity guarantees.
Moreover, dynamic settings have been explored in~\cite{pouryousef2023quantum}, addressing peak demands, aggregate rate or service delay, and in~\cite{le2022dqra}, aiming to satisfy the maximum number of user demands within deadlines.
However, entanglement routing from the perspective of fair and efficient resource allocation accounting for QoS perception of individual users remains largely unexplored, a gap this work seeks to address. 
We specifically study the single-path utility-based entanglement routing problem, where each demand is routed along a single path.

% In a naive approach, the single-path utility-based routing problem can be solved by addressing multiple instances of the QNUM problem, each corresponding to a feasible routing for given demands, and selecting the one with the highest network utility. 
% However, the number of feasible routes for a given demand grows exponentially with the number of network nodes $n$, rendering this approach computationally prohibitive.
% To address this, we make the following contributions:
Given demands, a naive approach would solve the single-path utility-based routing problem by evaluating the QNUM problem for each feasible routing and selecting the one with the highest network utility. 
However, the number of feasible routes grows exponentially with the number of network nodes $n$, rendering this approach computationally prohibitive. 
To address this, we make the following contributions:
\begin{itemize}
    \item \looseness = -1 We formulate the utility-based routing problem as a mixed-integer convex program (MICP) with $2kl$ binary and $6kl\!+\!k\!+\!l$ variables in total, where $k$ and $l$ denote the number of demands and links, respectively. Following~\cite{vardoyan2023quantum}, we consider Secret Key Fraction (SKF), Distillable Entanglement (DE), and negativity as entanglement measures~\cite{plenio2005introduction}, which quantify \textit{user satisfaction} for an allocated fidelity level. The MICP is exact for negativity, or when the network supports sufficiently high entanglement generation rates per demand. For other cases, we propose a method to bound the approximation error, which is seen to be below $0.001\%$ on considered real-world examples.
    \item While the MICP remains tractable for networks of moderate size, we propose to use its convex relaxation for (i) deriving an upper bound for the maximum achievable network utility and (ii) a randomized heuristic based on the idea of randomized rounding~\cite{raghavan1987randomized} for large networks.
    \item \looseness = -1 We further introduce a min-congestion-based~\cite{raghavan1987randomized} randomized heuristic, which is computationally faster, and an associated upper bound.
    On evaluated real-world examples, the second heuristic consistently yields superior \textit{average} performance, while the upper bound is often closer to the optimum than its counterpart.
    In practice, the better of the two can be selected for new instances.
    % This heuristic is computationally faster than the previous heuristic, while the upper bound calculation is faster compared to the original MICP. 
\end{itemize}  
The rest of the paper is structured as follows. 
Sect.~\ref{sec:prelims} introduces the network model and entangled link generation process. 
The general problem formulation is presented in Sect.~\ref{sec:assumptions}, followed by the direct link-based and min-congestion formulations in Sect.~\ref{sec:directSinglePath} and~\ref{sec:minCong}, respectively. 
Finally, Sect.~\ref{sec:eval} presents numerical evaluations on real-world networks. 
\section{Preliminaries}\label{sec:prelims}
% \vspace{-3pt}
\everypar{\looseness=-1}

% In this work, we consider a repeater network where the end nodes represent users and the intermediate nodes represent repeaters.
% Given a set of demands, represented by pairs of end nodes, our objective is to allocate resources—specifically, the entanglement generation rate and quality—across all possible routes connecting these node pairs, with the aim of maximising the network utility~\cite{kar2024convexification,vardoyan2023quantum}. 
% A key distinction between~\cite{kar2024convexification,vardoyan2023quantum} and our approach is that the selection of routes between the specified SD pairs is not predetermined. 
% Before formally presenting the utility-based entanglement routing problem, 
In this section, we provide a brief overview of the network structure and the state description of the quantum communication links within the network.
Notation-wise, we use $[m]$ to denote the set $\{1,2,\dots,m\}$ for $m \!\in\! \mathbb{N}$, whereas $B_j$ and $B_{:i}$ respectively denote the $j$th row and $i$th column of a matrix $B$.

\subsection{Network and demands}\label{sec:demands}
\everypar{\looseness=-1}

We represent the network topology as a graph ${G\!:=\!(V,E)}$, where $V$ denotes the set of nodes and $E$ denotes the set of edges.
The end nodes in $V$ correspond to users, and the edges in $E$ represent the direct quantum communication links between adjacent nodes.
We assume that we are given a set of demands $D$, comprising $k$ source-destination (SD) pairs
\begingroup
\setlength{\abovedisplayskip}{5pt}   % space above equations
\setlength{\belowdisplayskip}{5pt}   % space below equations
\setlength{\abovedisplayshortskip}{5pt}
\setlength{\belowdisplayshortskip}{5pt}
\begin{align}\label{eq:demands}
    D:=\{(s_i,t_i):s_i,t_i \in V,\, i \in [k]\}\,.
\end{align}
\endgroup
\looseness = -1 
% Here, $[k]:=\{1,2,\dots,k\}$.
A route is defined as a path in $G$ between the corresponding SD pair.
We restrict routes to \textit{simple} paths for each SD pair, a choice we will later show entails no loss of generality.
We assume that demand $(s_i,t_i)$ (alternatively, demand $i$) has $p_i$ many routes.
We index the routes sequentially, and the set of route indices serving demand $i$ is denoted as $I_i$.
% i.e., the set of route indices serving $(s_i,t_i)$ is given by $[P_i]\!\setminus\! [P_{i-1}]$, where ${P_i :=\sum_{m \in [i]} p_m}$ for $i \in [k]$ and $[P_0]:=\emptyset$.
Further, we prune the edges and vertices that are not incident on any route.
We update the notations $G,V,E$ to denote this pruned graph and assume that $G$ has $r$ routes and $l$ links in total.
That is,
\begingroup
\setlength{\abovedisplayskip}{5pt}   % space above equations
\setlength{\belowdisplayskip}{5pt}   % space below equations
\setlength{\abovedisplayshortskip}{2pt}
\setlength{\belowdisplayshortskip}{2pt}
\begin{align*}
    r = \sum\nolimits_{i \in [k]} p_i~, \quad l = |E|~.   
\end{align*}
\endgroup

\subsection{Link Generation and State Description}\label{sec:singleClick}
\everypar{\looseness=-1}

In our network graph, each edge represents a direct quantum communication link between adjacent users/repeaters.
We assume that entanglements in these elementary links are produced using the single-click protocol~\cite{cabrillo1999creation}, where the generated states have the following form:
\begin{align}\label{eq:singlePhoton}
    \rho = (1 - \alpha) \lvert \Psi^+ \rangle \langle \Psi^+ \rvert +  \alpha \lvert \uparrow \uparrow \rangle \langle \uparrow \uparrow \rvert~. 
\end{align}
Here $\alpha$ denotes the bright-state
population and $\ket{\Psi^+}$ is a Bell state orthogonal to the bright state $\ket{\uparrow \uparrow}$.
The probability of success of each generation attempt is given by 
\begingroup
\setlength{\abovedisplayskip}{5pt}   % space above equations
\setlength{\belowdisplayskip}{5pt}   % space below equations
\setlength{\abovedisplayshortskip}{2pt}
\setlength{\belowdisplayshortskip}{2pt}
\begin{align*}
    p_{\text{elem}} = 2 \kappa \eta \alpha~,
\end{align*}
\endgroup
where $\kappa \in (0,1)$ is a multiplicative factor accounting for the inefficiencies other than photon loss in the fiber. 
Further, $\eta$ denotes the transmissivity of the link from one end to the midpoint heralding station.
For a link of length $L$ km, its transmissivity can be calculated as $\eta = 10^{-0.02 L}$.

Motivated by the mathematical convenience of handling Werner states and the fact that any bipartite state can be transformed into a Werner state of the same fidelity 
% through the application of random transformations from a set of operations involving identical rotations on each qubit
~\cite{dur2007entanglement,bennett1996mixed}, we assume that the elementary links generated as~\eqref{eq:singlePhoton} are further converted to Werner states. 
Accordingly, the link state can be described as
\begingroup
\setlength{\abovedisplayskip}{5pt}   % space above equations
\setlength{\belowdisplayskip}{5pt}   % space below equations
\setlength{\abovedisplayshortskip}{2pt}
\setlength{\belowdisplayshortskip}{2pt}
\begin{align}\label{eq:werner}
    \rho_w = w \lvert \Psi^+ \rangle \langle \Psi^+ \rvert +  (1-w)~ \nicefrac{\mathbb{I}_4}{4}~. 
\end{align}
\endgroup
For the states in~\eqref{eq:singlePhoton} and~\eqref{eq:werner} to have same fidelity, we must have
\begingroup
\setlength{\abovedisplayskip}{5pt}   % space above equations
\setlength{\belowdisplayskip}{5pt}   % space below equations
\setlength{\abovedisplayshortskip}{2pt}
\setlength{\belowdisplayshortskip}{2pt}
\begin{align*}
    \vspace{0pt}1-\alpha = \frac{1+3w}{4}~, \quad \text{i.e.,} \quad \alpha = \frac{3(1-w)}{4}~.
\end{align*}
% We assume that $w$ takes values in its theoretically possible range $[0,1]$, i.e., we do not impose an additional cutoff upfront.
% As mentioned earlier, the state description~\eqref{eq:werner} also leads to a convenient Werner form for e2e entanglements, created by swapping the elementary links on a given route (see~\ref{A:entanglementSwap}).
%
If entanglement generation is attempted every $T$ unit of time, the resulting generation rate can be expressed as
\begingroup
\setlength{\abovedisplayskip}{5pt}   % space above equations
\setlength{\belowdisplayskip}{5pt}   % space below equations
\setlength{\abovedisplayshortskip}{5pt}
\setlength{\belowdisplayshortskip}{5pt}
\begin{align}\label{eq:genRate}
    \frac{p_{\text{elem}}}{T} = d(1 - w) ~\,\text{with} ~\, d:=\frac{3 \kappa \eta}{2T}\,,
\end{align}
\endgroup
which provides the effective rate-fidelity trade-off in elementary link generation.

\section{Assumptions and Formulation}\label{sec:assumptions}
\everypar{\looseness=-1}
% While our main focus is the single-path version of the utility-based routing problem, we first formulate the problem in its full generality in Sect.\ref{sec:multipath}, allowing for multipath routing and any entanglement measure~\cite{plenio2005introduction}.
% We then highlight the intractability of the general problem and present solution approaches for the single-path case in Sect.\ref{sec:singlePath}.
To describe the problem formally, we make the following assumptions, many of which are similar to that of~\cite{vardoyan2023quantum,kar2024convexification}.

\begin{enumerate}[label={A\arabic*},nolistsep,leftmargin=*]

\item \label{A:staticNet} \looseness = -1 \textit{Static network:} 
Analogous to classical NUM and QNUM, 
we aim to distribute communication resources
% , i.e., entanglement generation rate and quality, among demands 
in a static setting.
Specifically, routes and their corresponding rate-fidelity allocations are determined before the network goes into operation. 
We assume that each route supports a single application throughout the operational period, implying the \textit{derived value} of an allocation remains the same for a demand over time.
The network topology is also assumed to be fixed.

\item \label{A:entanglementSwap} \textit{Entanglement swapping:} 
End-to-end entanglement between the terminal nodes is established via a two-step process.  
First, entanglement is generated at the link level between adjacent nodes.  
Then, entanglement swapping is performed at intermediate repeater nodes along the path to produce e2e entanglement~\cite{pan1998experimental}.  
Following~\cite{vardoyan2023quantum,kar2024convexification}, we assume the simplified setting where link-level entanglements are generated simultaneously and swapped immediately, avoiding decoherence.  
Since swapping two Werner states yields another Werner state whose parameter equals the product of the initial parameters~\cite{munro2015inside}, this  enables a concise representation of the e2e link.

\item \label{A:attemptFreq} \looseness = -1 \textit{Rate and fidelity of entanglement generation:}  As described in Sect.~\ref{sec:singleClick}, the rate-fidelity trade-off of the elementary link generation is inspired by the single-click protocol and follows the realtion~\eqref{eq:genRate}.
That is, if we fix the fidelity of link $j \in [l]$ by setting its Werner parameter to $w_j$, the maximum entanglement generation rate will be $\mu_j \!:=\! d_j(1\!-\!w_j)$, where $d_j \!:=\! 3 \kappa_j \eta_j/ 2T_j$~\eqref{eq:genRate}.
Note that~\ref{A:staticNet} implies that $w_j$'s are set in advance and remain fixed throughout network operation.
Consequently, the contributions of the $j$th link towards the e2e Werner parameters are identical for all routes traversing it.
% Here, $\mu_j$ can be interpreted as the capacity of link $j$ when it produces Werner states with parameter $w_j$.

% \item \label{A:fixedFidelity} \textit{Arbitrary but fixed quality of entanglement:} The Werner parameter of a link-level entanglement ${w_j, j\in [l]}$ can be chosen arbitrarily for optimising the network utility but remains fixed once chosen.
% In other words, we perform a one-shot analysis, where the values of $w_j$'s are set in advance and cannot be adjusted dynamically during the operation of the network, as the routes are served concurrently.
% This also implies that the contributions of the $j$-th link towards the e2e Werner parameters are the same across the routes passing through that link.
% Observe from~\eqref{eq:genRate} that increasing the value of the Werner parameter reduces the entanglement generation rate. 
%
\item \label{A:utility} \looseness = -1 \textit{Utility of a demand:} 
We first define the overall (binary) link-route incidence matrix $A$, where $a_{jm}\!:=\!((A))_{jm}\!=\!1$ iff the $m$th route traverses link $j$.
Recall that we restrict our choice of routes to simple paths.
The link-route incidence matrices corresponding to valid single-path routings of the $k$ demands belong to the following set
\begingroup
\setlength{\abovedisplayskip}{5pt}   % space above equations
\setlength{\belowdisplayskip}{5pt}   % space below equations
\setlength{\abovedisplayshortskip}{2pt}
\setlength{\belowdisplayshortskip}{2pt}
\begin{flalign}\label{eq:setRoutes}
    \mathcal{P}(A) := \{\Tilde{A} \!\in \!\{0,1\}^{l \times k}\!: \Tilde{A}_{:i} \!=\! A_{:m}\,\text{for some}\,m\! \in \!I_i \}
\end{flalign}
\endgroup
Since each demand is served via a single route, we denote the rate allocated to demand $i$ by $x_i$ and the e2e Werner parameter by $u_i$.
As the total rate allocated to demands cannot exceed a link's maximum entanglement generation rate for any valid single-path routing $\tilde{A}$, we must have ${\sum_{i \in [k]}\tilde{a}_{ji} x_i \le \mu_j}$, where $\Tilde{a}_{ji} \!=\! ((\tilde{A}))_{ji}$.
Recall from~\ref{A:entanglementSwap} that the e2e Werner parameter is the product of link-level Werner parameters.
Given a routing $\Tilde{A}$, we then have ${u_i \!=\! \prod_{j \in [l]}w_j^{\Tilde{a}_{ji}}}$.
% If route $m$ serves demand~$i$, 
We quantify the suitability of the fidelity (e2e Werner parameter to be precise) allocation $u_i$ using a nonnegative non-decreasing function $f_i$, where 
\begingroup
\setlength{\abovedisplayskip}{5pt}   % space above equations
\setlength{\belowdisplayskip}{5pt}   % space below equations
\setlength{\abovedisplayshortskip}{2pt}
\setlength{\belowdisplayshortskip}{2pt}
\begin{align*}
    f_i:&[0,1] \to [0,\beta_i] 
    \\    & u_m \mapsto f_i(u_m)
\end{align*}
\endgroup

In the QNUM framework, $f_i$'s are generally taken to be entanglement measures~\cite{plenio2005introduction}, secret key fraction, or fidelity of teleportation.
% although fidelity  
% Note that $f_i$'s are assumed to be non-decreasing, and
Following~\cite{vardoyan2023quantum}, the utility of a demand is assumed to have the form $x_i f_i(u_i)$.
Importantly, the demand utility is a function of the routing $\tilde{A}$.
% Finally, the network utility is formulated as the product of the route utilities. 
% The product form of individual and network utilities ensures that a specified level of network utility is achieved only when each route is allocated both rate and fidelity adequately.
% Note that it is possible to have different forms for route and network utility functions than ours.

\item \label{A:networkUtility} \looseness = -1 \textit{Network utility:} 
% In the general multipath setting, a demand can be served by multiple routes.
% In order to define network utility, we first define the total utility of the $i$th demand $U_i$ as the sum of the corresponding route utilities:
% \begin{align}\label{eq:demandUtil}
%     U_i := \sum_{m \in I_i} x_m f_i(u_m)~,
% \end{align}
% %
% where $I_i$ denote the set of route indices serving demand~$i$.
For a given single-path routing $\tilde{A}$, the network utility $\mathcal{U}(\tilde{A})$ is defined as the product of the demand utilities:
% However, unlike the QNUM problem, the network utility is dependent on the choice of routing.
 % we denote the rate allocation 
 \begingroup
 \setlength{\abovedisplayskip}{5pt}   % space above equations
 \setlength{\belowdisplayskip}{5pt}   % space below equations
 \setlength{\abovedisplayshortskip}{2pt}
 \setlength{\belowdisplayshortskip}{2pt}
\begin{align}\label{eq:netUtil}
    \mathcal{U}(\tilde{A}) := \prod_{i \in [k]} x_i f_i\Big(\prod_{j \in [l]}w_j^{\Tilde{a}_{ji}}\Big).
\end{align}
\endgroup
Note that alternative forms for demand and network utility functions are possible.
However, the product form guarantees that a given network utility level is reached only when each demand receives sufficient rate and fidelity. 
In canonical NUM, the utility is typically written as $\ln \mathcal{U} = \sum_i \ln U_i$, where $U_i$ is the utility of demand $i$. 
In contrast, we follow the welfare economics convention \cite{johansson1991introduction}, from which NUM originates, to motivate our formulation. 
Both approaches are clearly equivalent for the purpose of utility maximization. 
\end{enumerate}
Based on these assumptions, we now formulate the utility-based entanglement routing problem.
A list of recurring notations is provided in Tab.~\ref{tab:defs} in the appendix.
% % \vspace{-4pt}
% A list of parameters describing the network and the auxiliary variables are given in Tab.~\ref{tab:defs}.

\subsection{The Utility-based Entanglement Routing Problem}\label{sec:multipath}
\everypar{\looseness=-1}
We denote the rate allocation vector for the demands by ${\vec{x} \!=\! (x_1,x_2,\dots,x_k)}$ and the Werner parameter vector for the links by ${\vec{w} \!=\! (w_1,w_2,\dots,w_l)}$.
Also, let $\Tilde{A}_j$ denote the $j$th row of routing $\tilde{A}$.
The single-path entanglement routing problem can then be written as
 \begingroup
\setlength{\abovedisplayskip}{5pt}   % space above equations
\setlength{\belowdisplayskip}{5pt}   % space below equations
\setlength{\abovedisplayshortskip}{2pt}
\setlength{\belowdisplayshortskip}{2pt}
\begin{align}\label{eq:uRouteCanonical}%% \vspace{-0pt}
  \max\limits_{\vec{x},\vec{w}, \tilde{A}} \quad & \prod_{i \in [k]} x_i f_i\bigg(\prod\limits_{j=1}^l w_j^{\tilde{a}_{ji}}\bigg)\Bigg) \nonumber\\
  \text{s.t.} \quad & \vec{0} \preceq \vec{x},\thickspace \text{(Non-negative rates)} \\
  % \quad & 0 < \sum_{m \in I_i}  x_m f_i\bigg(\prod\limits_{j=1}^l w_j^{a_{jm}}\bigg) \quad \forall i \in [k] \\
  %& \quad \text{(positive demand utility)} \nonumber\\
  & \vec{0} \preceq \vec{w} \preceq \vec{1},\thickspace \text{(Fidelity bounds)} \nonumber\\
  &\langle \tilde{A}_j,\vec{x} \rangle \!\le\! \mu_j \!=\! d_j(1\!-\!w_j) \thinspace~ \forall j \!\in\! [l],\thickspace \text{(Rate constraints)}\nonumber \\
  & \tilde{A} \in \mathcal{P}(A) .\thickspace \text{(Single path constraint)}\nonumber 
\end{align}
\endgroup
Here, $\preceq$ denotes element-wise inequality,
and $\langle \tilde{A}_j,\vec{x} \rangle$ denotes the dot product of $\tilde{A}_j$ and $\vec{x}$.
As argued in~\cite{kar2024convexification}, the monotonicity assumption on $f_i$'s lets us replace the inequality in the rate constraints in~\eqref{eq:uRouteCanonical} by the equality ${w_j \!=\! 1\!-\!\langle \tilde{A}_j,\vec{x} \rangle/d_j}$ and thereby eliminate $\vec{w}$:
% % \vspace{-5pt}
% \begin{align}\label{eq:uRouteTransformedInt}%% \vspace{-0pt}
%   \max\limits_{\vec{x}} \quad & \prod_{i \in [k]} \Bigg(\sum_{m \in I_i}  x_m f_i\bigg(\prod\limits_{j=1}^l \bigg( 1\!-\!\frac{\langle A_j,\vec{x} \rangle}{d_j}\bigg)^{a_{jm}}\bigg)\Bigg) \nonumber\\
%   \text{s.t.} \quad & \vec{0} \preceq \vec{x}~,\\
%   \quad & \langle A_j,\vec{x} \rangle \le d_j \quad \forall j \in [l]~.\nonumber
% \end{align}

% \looseness = -1
% \begin{remark}
%     The multupath entanglement routing problem~\eqref{eq:uRouteTransformedInt} in general has an exponential number of variables ($x_m$) in terms of the size of the network $n$ (i.e., $|V|$) as there are, in general, an exponential number of routes between a given SD pair.
%     In the related topic of multi-commodity flow problems, this issue is addressed by switching to a link-based formulation.
%     However, for the multipath entanglement routing problem, expressing the network utility in terms of link-level variables is non-trivial and we focus on the single-path version instead.
% \end{remark}
% % \vspace{-7pt}

% \subsection{Single-path Entanglement Routing}\label{sec:singlePath}

% \looseness = -1 We first provide a natural route-based formulation of the single-path version.
% We observe that 
% % \vspace{-0pt}
% \noindent where $\Tilde{A}_{:i}$ denotes the $i$th column of $\Tilde{A}$.
% Denoting $\Tilde{a}_{ji} := ((\Tilde{A}))_{ji}$ and  
% updating $\vec{x}$ as the \textit{resized} rate allocation vector ${\vec{x} = \!=\! (x_1,x_2,\dots,x_k)}$, the single-path routing problem can be expressed as

\begingroup
\setlength{\abovedisplayskip}{5pt}   % space above equations
\setlength{\belowdisplayskip}{5pt}   % space below equations
\setlength{\abovedisplayshortskip}{2pt}
\setlength{\belowdisplayshortskip}{2pt}
\begin{align}
  \max\limits_{\vec{x},\Tilde{A}} \quad & \prod_{i \in [k]} x_i f_i\bigg(\prod\limits_{j=1}^l \bigg( 1\!-\!\frac{\langle \tilde{A}_j,\vec{x} \rangle}{d_j}\bigg)^{\tilde{a}_{ji}}\bigg) \label{eq:uRouteSingleWd}\\
  \text{s.t.} \quad & \vec{0} \preceq \vec{x}~,\\
  & \langle \tilde{A}_j,\vec{x} \rangle \le d_j \quad \forall j \in [l]~, \label{eq:uRouteSingleWdMid}\\
  & \tilde{A} \in \mathcal{P}(A)~.\label{eq:uRouteSingleWdEnd}   
\end{align}
\endgroup
% \begingroup
% \setlength{\abovedisplayskip}{0pt}   % space above equations
% \setlength{\belowdisplayskip}{0pt}   % space below equations
% \setlength{\abovedisplayshortskip}{0pt}
% \setlength{\belowdisplayshortskip}{0pt}
% \begin{proposition}
%     The single-path entanglement routing problem~\eqref{eq:uRouteSingleWd}--\eqref{eq:uRouteSingleWdEnd} is in general NP-hard.
% \end{proposition}
% \endgroup

\begin{remark}
	Note that on a non-simple path, removal of the loops increases the value of the e2e Werner parameter and, by monotonicity of entanglement measures $f_i$, the value of the objective function~\eqref{eq:uRouteSingleWd} as well. Thus, restricting $A$ to simple paths can be done without loss of generality.
\end{remark}
\begin{remark}
    Entanglement measures require $f_i(u)\!=\!0$ for $u \!\le\! 1/3$ as Werner states are separable until the threshold Werner parameter value of $1/3$.
    However, if we only require $f_i$'s to be non-decreasing and non-negative, for the special case ${f_i(u)\!=\!1}$ and ${d_j\!=\!1}\, \forall i,j$, the objective function in~\eqref{eq:uRouteSingleWd} attains the value $1$ iff there exist edge-disjoint paths (EDP) for the given demands.
    Since the EDP problem is NP-hard~\cite{guruswami1999near}, this implies that under this relaxed assumption, the single-path entanglement routing problem cannot be solved efficiently.
\end{remark}
% % \vspace{-6pt}

\looseness=-1 For a given demand, the number of paths is in general exponential in the number of nodes $n$ (i.e., $|V|$), which implies that the number of candidate routings $\tilde{A}$ in~\eqref{eq:uRouteSingleWdEnd} grows exponentially with $n$.
To address this, we present a link-based formulation of~\eqref{eq:uRouteSingleWd}--\eqref{eq:uRouteSingleWdEnd}.
For certain entanglement measures $f_i$, the problem reduces to an MICP with $6kl\!+\!k\!+\!l$ variables (of which $2kl$ are binary) or can be \textit{closely} approximated by it.
% This contrasts with~\eqref{eq:uRouteSingleWdEnd}, where the number of feasible routings $\tilde{A}$ grows exponentially with the number of nodes $n$.
% % \vspace{-3pt}

\vspace{-10pt}
\section{Link-based Formulation 
}\label{sec:linkBased}
\vspace{-5pt}
\everypar{\looseness=-1} 

\subsection{Direct Formulation}\label{sec:directSinglePath}

To formulate the optimization problem solely in terms of demand-based link-level variables, we first convert the undirected network graph into a directed one by replacing each undirected edge with two directed edges, which are then re-labelled.  
We denote the indices of the directed edges (communication links) in the new graph by $j'$, in contrast to $j$ in the undirected graph.
We now introduce the variables required to define the single-path routing problem.
\begin{center}
% % \vspace{-3pt}
\begin{tabular}{p{0.1\columnwidth}p{0.8\columnwidth}}
  $\epsilon(j)$ & the set comprising two directed link indices corresponding to the undirected link $j$, $j \in [l]$ \\
  $\delta^+(v)$ & the set of indices of incoming links at $v$, $v \in V$ \\
  $\delta^-(v)$ & the index set of outgoing links from $v$, $v \in V$ \\
  % $\delta(v)$ & $\delta^+(v) \cup \delta^-(v)$ \\
  $x_{ij'}$ & the rate allocated to demand $i$ on the $j'$th directed link, $i \in [k], j' \in [2l]$ 
\end{tabular}
% % \vspace{-5pt}
\end{center}
%   
% Note that for indexing, we can use \textit{demand} and \textit{route} interchangeably in the single-path case.
% Also, 
We group the rate allocation variables into the following allocation vector  
\begingroup
\setlength{\abovedisplayskip}{5pt}   % space above equations
\setlength{\belowdisplayskip}{5pt}   % space below equations
\setlength{\abovedisplayshortskip}{2pt}
\setlength{\belowdisplayshortskip}{2pt}
\begin{align}
  \overset{\leftrightarrow}{x} := (x_{11},x_{12},\dots,x_{1\, 2l},\dots,x_{k1},x_{k2},\dots,x_{k\, 2l})\,.  
\end{align}
\endgroup

\looseness=-1 Now, the rate allocated to demand $i$ is given by the maximum allocation to this demand over links originating from its source node~$s_i$, i.e., $\max_{j' \in \delta^{-}(s_i)} x_{ij'}$.  
Also, the $(j,i)$th element $\tilde{a}_{ji}$ of the routing matrix $\tilde{A}$ is $1$ iff there is an allocation on this link in either direction, i.e., ${\sum_{j' \in \epsilon(j)}x_{ij'} \!>\! 0}$. 
We can thus rewrite the contribution of (undirected) link $j$ to the e2e Werner parameter of the route serving demand~$i$ in~\eqref{eq:uRouteSingleWd} as

\begingroup
\setlength{\abovedisplayskip}{5pt}   % space above equations
\setlength{\belowdisplayskip}{5pt}   % space below equations
\setlength{\abovedisplayshortskip}{2pt}
\setlength{\belowdisplayshortskip}{2pt}
\begin{align}\label{eq:wernerBidir}
  w_{ij}(\overset{\leftrightarrow}{x}) \!:=\! \bigg( 1\!-\!\frac{\langle \tilde{A}_j,\vec{x} \rangle}{d_j}\bigg)^{\tilde{a}_{ji}} \!\!\!\!=\! 1 \!-\!\mathbbm{1}_{\{ \!\!\sum\limits_{j' \in \epsilon(j)} \!\!x_{ij'} > 0\}}\!\!\!\! \sum_{i \in [k],\,j' \in\epsilon(j)} \!\!\!\!\!\!\nicefrac{x_{ij'}}{d_j}\,,
\end{align}
\endgroup
where $\mathbbm{1}_{\{\}}$ denotes the indicator function.
This gives the following link-based formulation of the single-path routing problem:

\begingroup
\setlength{\abovedisplayskip}{5pt}   % space above equations
\setlength{\belowdisplayskip}{5pt}   % space below equations
\setlength{\abovedisplayshortskip}{2pt}
\setlength{\belowdisplayshortskip}{2pt}
{\small
\begin{align}
\max_{\overset{\leftrightarrow}{x}}~~ & \prod_{i \in [k]} \!\! \Big( \max_{j' \in \delta^{-}(s_i)}\! x_{ij'} \!\Big) 
  f_i\bigg( \prod_{j = 1}^{l} w_{ij}\big(\overset{\leftrightarrow}{x}\big) \bigg) \label{eq:objDirect}\\
\text{s.t.} \quad 
& x_{ij'} \geq 0, ~~ \forall i \in [k], \; \forall j' \in [2l] \label{eq:nonNegDirect}\\
& \sum_{i \in [k], \, j' \in \epsilon(j)} x_{ij'} \le d_j, ~~ \forall j \in [l] \label{eq:capacityDirect}\\
& \sum_{j' \in \delta^{+}(v)} \!\!\!\!\mathbbm{1}_{\{x_{ij'}>0\}} \!\le\! 1, \,\forall v \!\in\! V, \!\! \sum_{j' \in \delta^{+}(s_i)} \!\!\!\!\!\mathbbm{1}_{\{x_{ij'}>0\}} \!=\! 0; \, \forall i \!\in\! [k] \label{eq:single1Direct}\\
& \sum_{j' \in \delta^{-}(v)} \!\!\!\!\mathbbm{1}_{\{x_{ij'}>0\}} \!\le\! 1, \,\forall v \!\in\! V, \!\! \sum_{j' \in \delta^{-}(t_i)} \!\!\!\!\!\mathbbm{1}_{\{x_{ij'}>0\}} \!=\! 0; \, \forall i \!\in\! [k] \label{eq:single2Direct}\\
& \sum_{j' \in \delta^{+}(v)}\!\!\!\! x_{ij'} -\!\!\!\! \sum_{j' \in \delta^{-}(v)}\!\!\!\! x_{ij'} \!=\! 0, 
~\forall v \!\notin\! \{s_i,t_i\}, \forall i \!\in\! [k] \label{eq:flowDirect}
\end{align}
}
\endgroup

We now explain the constraints:
\begin{itemize}
    \item \textit{Positivity}: rate allocations must be nonnegative~\eqref{eq:nonNegDirect}.
    \item \textit{Capacity constraints}: the total bidirectional allocation on the $j$th (undirected) link cannot exceed $d_j$~\eqref{eq:capacityDirect} (see~\eqref{eq:uRouteSingleWdMid}).
    \looseness = -1 \item \textit{Single path}: for each demand, at most one incoming link (resp. outgoing) of a vertex can have a positive rate allocation~\eqref{eq:single1Direct} (resp.~\eqref{eq:single2Direct}), while there are no allocations on the incoming (resp. outgoing) links to the source (resp. destination). This also ensures the absence of loops.
    \item \textit{Flow preservation}: for each demand, the incoming and outgoing rates must be equal at each node, except for the corresponding source and destination~\eqref{eq:flowDirect}.
\end{itemize}

\begin{proposition}
	The route-based~\eqref{eq:uRouteSingleWd}--\eqref{eq:uRouteSingleWdEnd} \textnormal{(\textbf{R})} and link-based~\eqref{eq:objDirect}--\eqref{eq:flowDirect} \textnormal{(\textbf{L})} formulations are equivalent.
\end{proposition}
\begin{proof}[Proof]
    We need to show that any feasible solution of \textnormal{(\textbf{R})} corresponds to a feasible solution of \textnormal{(\textbf{L})} with the same objective value, and vice versa.
    The proof follows the standard technique for converting route-based formulation to path-based formulation in multicommodity flow problems.
    Given a feasible $(\vec{x},\tilde{A})$ for \textnormal{(\textbf{R})}, each demand $i$ is assigned a unique route $\tilde{A}_{:i}$ with allocation $x_i$. 
    This route decomposes into a sequence of adjacent directed links from $s_i$ to $t_i$. 
    Setting $x_{ij'} \!= \!x_i$ on these links and ${x_{ij'} \!=\! 0}$ elsewhere satisfies the positivity~\eqref{eq:nonNegDirect}, single-path~\eqref{eq:single1Direct}--\eqref{eq:single2Direct} and flow preservation~\eqref{eq:flowDirect} constraints by definition and the capacity constraint~\eqref{eq:capacityDirect} due to the capacity constraint~\eqref{eq:uRouteSingleWdMid} of (\textbf{R}).
    Also, the objective values match by the argument preceding~\eqref{eq:wernerBidir}.
    
    \looseness = -1 Conversely, for any feasible solution of \textnormal{(\textbf{L})}, constraint~\eqref{eq:single2Direct} implies that for each demand $i$, $x_{ij'}$ can be positive on at most one outgoing link from $s_i$. 
    If all outgoing allocations from $s_i$ are zero, any simple path from $s_i$ to $t_i$ may be selected, yielding a corresponding feasible solution of \textnormal{(\textbf{R})} with zero objective value in both formulations.
    \textit{Otherwise}, flow conservation~\eqref{eq:flowDirect} together with the constraint that, for each demand \( i \), \( x_{ij'} \) is positive on at most one outgoing link from every node and there is no outgoing flow from \( t_i \)~\eqref{eq:single2Direct}, implies that flows cannot split and that there is a unique route from \( s_i \) to \( t_i \) with positive allocation.
    Further, due to the constraint that \( x_{ij'} \) is positive on at most one incoming link to every node, with no incoming flow to \( s_i \)~\eqref{eq:single1Direct}, this route corresponds to a simple path.
    We have thus found a unique directed simple path \( \tilde{A}_{:i} \) from \( s_i \) to \( t_i \) and we set $x_i$ to the common nonzero value of $x_{ij'}$ along this path. 
    The equality of the objective values again follows from the discussion preceding~\eqref{eq:wernerBidir}.
\end{proof}

% We now update formulation (\textbf{L}) to make it an MICP.
We now update formulation (\textbf{L}) to obtain the following MICP. 
The required steps are explained subsequently.

\noindent
\fbox{%
  \begin{minipage}{\dimexpr\columnwidth-2\fboxsep-2\fboxrule\relax}
  \small
  \begin{align}
\min_{\substack{\overset{\leftrightarrow}{x},\, \vec{y},\,\vec{\sigma},\\ \vec{\gamma},\, \vec{v},\, \vec{z}}}  &\sum_{i \in [k]} \!\Big(\lambda_i \!-\!
  \hat{F}_i (z_i)\Big) \label{eq:modObjective}\\
\text{s.t.} \quad 
& \vec{0} \preceq \overset{\leftrightarrow}{x}\,, \\
& \!\!\! \sum_{j' \in \epsilon(j),\, i \in [k]} x_{ij'} \le d_j,\, \forall j \in [l] \\
& \sum_{j' \in \delta^{+}(v)}\!\!\! x_{ij'} = \!\!\! \sum_{j' \in \delta^{-}(v)}\!\!\! x_{ij'}, 
\, \forall v \!\notin\! \{s_i,t_i\}, \, \forall i \!\in\! [k] \\
& y_{ij'} \in \{0, 1\},\, \forall j' \!\in\! [2l],\, \forall i \!\in\! [k] \label{eq:yBool}\\
& x_{ij'} \le d_j y_{ij'},\, \epsilon(j) \ni j'; ~ \forall j' \in [2l],\,\forall i \!\in\! [k] \\
& x_{ij'} \ge \varepsilon y_{ij'},\, \forall j' \in [2l],\,\forall i \!\in\! [k] \\
&\sum_{j' \in \delta^{+}(v)} y_{ij'} \le 1, \,\forall v \!\in\! V\!\setminus\!\{s_i\} , \, \forall i \!\in\! [k] \label{eq:yIncoming}\\
&\sum_{j' \in \delta^{-}(v)} y_{ij'} \le 1, \,\forall v \!\in\! V\!\setminus\!\{t_i\} , \, \forall i \!\in\! [k] \label{eq:yOutgoing}\\
& \sum_{j' \in \delta^{-}(s_i)} y_{ij'} = 1, \sum_{j' \in \delta^{+}(s_i)} y_{ij'} = 0;\, \forall i \!\in\! [k] \label{eq:yOutgoingSource}\\
& \sum_{j' \in \delta^{-}(t_i)} y_{ij'} = 0, \sum_{j' \in \delta^{+}(t_i)} y_{ij'} = 1;\, \forall i \!\in\! [k] \label{eq:yIncomingSink}\\
& \sigma_{j} = \sum_{j' \in \epsilon(j),\, i \in [k]} \nicefrac{x_{ij'}}{d_j},\, \forall j \in  [l]\\
& \gamma_{ij} \le \sum_{j' \in \epsilon(j)} y_{ij'},\, \forall j \!\in [l]\,,\,\forall i \!\in\! [k] \\
& \gamma_{ij} \le \sigma_{j},\, \forall j \!\in [l]\,,\,\forall i \!\in\! [k] \\
& \gamma_{ij} \ge 0,\, \forall j \!\in [l]\,,\,\forall i \!\in\! [k] \\
& \gamma_{ij} \ge \sigma_{j}-1+\sum_{j' \in \epsilon(j)} y_{ij'},\, \forall j \!\in [l]\,,\,\forall i \!\in\! [k] \\
& v_{ij} \le \ln(1-\gamma_{ij}),\, \forall j \!\in [l]\,,\,\forall i \!\in\! [k] \\
& z_i = \sum_{j = 1}^{l} v_{ij},\, \forall i \in [k]\,. \label{eq:zMICP}
% \endgroup
\end{align}
  \end{minipage}%
}

\begin{itemize}
    \looseness = -1
   % \item \textit{Zero allocations}: the formulation produces solutions with simple paths for all demands since any loop in a candidate path can be removed to potentially improve the Werner parameter and thereby increase the objective value. 
    %Thus, for each demand, rate allocations on links (i)  incoming to the source node and (ii) outgoing from the destination node can be set to zero. 
    %For clarity of exposition, we do not explicitly remove these variables, although omitting them in practice improves computational efficiency.
    % Therefore, we drop the allocation variables $x_{ij'}$ for ${j' \!\in\! \delta^+(s_i)\cup \delta^-(t_i),\, i \!\in\! [k]}$ and \textit{reuse} the notation $\overset{\leftrightarrow}{x}$ for the updated allocation vector.
    % We also set $\delta^+(s_i)$ and $\delta^-(t_i)$ to $\emptyset$.
    % Moreover, we denote the set of demands with potentially nonzero allocation on the $j'$th link as $\Delta(j')$.
    % \imp{Max removal step is wrong, being updated. Except for Werner parameter in the next point, rest are wrong.}
    \item \textit{Surrogate variables for single path constraints}: we define
    \begin{align*}
        y_{ij'} = \mathbbm{1}_{\{x_{ij'}>0\}}\,,~\forall i \in [k],\, \forall j' \in [2l]\,,    
    \end{align*}
    and accordingly introduce the following constraint:
    \begin{align}\label{eq:yAsIndic}
      x_{ij'} \!\le\! d_j y_{ij'},\,x_{ij'} \!\ge\! \varepsilon y_{ij'},\,y_{ij'} \!\in\! \{0,1\},~\forall i,\,\forall j'\,,  
    \end{align}
    where $\varepsilon>0$ is arbitrarily small.
    Since the utility is zero when any demand lacks a positive allocation on at least one outgoing link from its source and one incoming link to its destination, we eliminate such cases by updating the single-path constraints~\eqref{eq:single1Direct}–\eqref{eq:single2Direct} to:
    \begingroup
    \setlength{\abovedisplayskip}{5pt}   % space above equations
    \setlength{\belowdisplayskip}{5pt}   % space below equations
    \setlength{\abovedisplayshortskip}{2pt}
    \setlength{\belowdisplayshortskip}{2pt}
    \begin{align}\label{eq:single1Direct_}
    \begin{aligned}
        &\sum_{j' \in \delta^{+}(v)} \!\!\!\!\!y_{ij'} \le 1,\, \forall v \!\in\! V\!\setminus\!\{s_i\}, \sum_{j' \in \delta^{+}(s_i)} \!\!\!\!\!y_{ij'}= 0;\, \forall i \!\in\! [k] \\
        &\sum_{j' \in \delta^{-}(v)} \!\!\!\!\!y_{ij'} \le 1,\, \forall v \!\in\! V\!\setminus\!\{t_i\}, \sum_{j' \in \delta^{-}(t_i)} \!\!\!\!\!y_{ij'}= 0;\, \forall i \!\in\! [k] \\
        &\sum_{j' \in \delta^{-}(s_i)} \!\!\! y_{ij'} = 1, \sum_{j' \in \delta^{+}(t_i)} \!\!\! y_{ij'} = 1;\, \forall i \!\in\! [k]
    \end{aligned}
    \end{align}
    \endgroup
    Also, absence of loops implies
    \begin{align}\label{eq:y+}
        y^+_{ij} := \sum_{j' \in \epsilon(j)} y_{ij'} \le 1,~\text{i.e.,}~ y_{ij}^+ \in \{0,1\}\, \forall i, j.
    \end{align}
    % % \vspace{-3pt}
    We group the binary surrogates into vectors $\vec{y}$ and $\vec{y}_+$.
    
    %for compactness.
    % \item \textit{Surrogate variables for max removal}: 
    % % constraints~\eqref{eq:single1Direct}, \eqref{eq:single2Direct}, \eqref{eq:flowDirect} imply
    % % $$\max_{j' \in \delta^{+}(v)} \!x_{ij'} = \max_{j' \in \delta^{-}(v)} \!x_{ij'}~\forall v \in V_0^c,\, \forall i \in [k]\,.$$
    % As a standard technique, we replace $\max_{j' \in \delta^{+}(v)} \!x_{ij'}$ and $\max_{j' \in \delta^{-}(v)} \!x_{ij'}$ by surrogate variables $y_{iv}$ for ${v \!\in\! V_0^c}$ and add constraints of the form ${y_{iv} \ge x_{ij'}}$ for all relevant $j'$s.
    % Also, $\max_{j' \in \delta^{-}(s_i)} \!x_{ij'}$ and $\max_{j' \in \delta^{+}(t_i)} \!x_{ij'}$ are respectively replaced by $y_{i s_i}$ and $y_{i t_i}$ for $i \in [k]$.
    \item \textit{Surrogate variables for Werner parameters}: we define
    \begingroup
    \setlength{\abovedisplayskip}{5pt}   % space above equations
    \setlength{\belowdisplayskip}{5pt}   % space below equations
    \setlength{\abovedisplayshortskip}{2pt}
    \setlength{\belowdisplayshortskip}{2pt}
    \begin{align}
        \sigma_{j} &:= \sum_{j' \in \epsilon(j),\, i \in [k]} \nicefrac{x_{ij'}}{d_j}\,,~ \forall j \in [l] \label{eq:constraintsWerner1}\\
        \gamma_{ij} &:= y^+_{ij}\, \sigma_{j}\,, ~\forall i \!\in [k]\,,\,\forall j \!\in\! [l] \label{eq:constraintsWerner2}\\
        v_{ij} &:=\ln(1-\gamma_{ij})\,, ~\forall i \!\in \![k]\,,\,\forall j \!\in\! [l] \label{eq:constraintsWerner3}\\
        z_i &:= \sum_{j = 1}^{l} v_{ij} \label{eq:constraintsWerner4}\,, ~\forall i \in [k]\,.
    \end{align}
    \endgroup
    That is, the Werner parameter $w_{ij}(\overset{\leftrightarrow}{x})$ from~\eqref{eq:wernerBidir} is reformulated as $e^{v_{ij}}$. 
    Eq.~\eqref{eq:constraintsWerner1} and ~\eqref{eq:constraintsWerner4} are linear constraints.
    We modify~\eqref{eq:constraintsWerner3} to 
    \begingroup
    \setlength{\abovedisplayskip}{5pt}   % space above equations
    \setlength{\belowdisplayskip}{5pt}   % space below equations
    \setlength{\abovedisplayshortskip}{2pt}
    \setlength{\belowdisplayshortskip}{2pt}
    \begin{align}\label{eq:constraintsWerner3.1}
      v_{ij} \le \ln(1-\gamma_{ij})\,, ~ \forall i \!\in \![k]\,,\,\forall j \!\in\! [l]
    \end{align}
    \endgroup
    which corresponds to a convex region due to concavity of the RHS.
    This does not relax the problem, as the objective function is non-decreasing in $v_{ij}$'s.
    Further, using ${0 \!\le\! \sigma_j \!\le\! 1}$ (capacity constraint) and~\eqref{eq:y+},
    we enforce~\eqref{eq:constraintsWerner2} via its exact McCormick relaxation~\cite{mccormick1976computability}:
    \begingroup
    \setlength{\abovedisplayskip}{5pt}   % space above equations
    \setlength{\belowdisplayskip}{5pt}   % space below equations
    \setlength{\abovedisplayshortskip}{2pt}
    \setlength{\belowdisplayshortskip}{2pt}
    \begin{align}
    \left. 
    \begin{aligned} 
     \gamma_{ij} &\le y_{ij}^+,\, \gamma_{ij} \ge 0\,, \\ 
     \gamma_{ij} &\le \sigma_{ij},\, \gamma_{ij} \ge \sigma_{ij} \!-\!(1\!-\!y_{ij}^+)
    \end{aligned} 
    \right\} ~\forall i \!\in\! [k],\,\forall j \!\in\! [l]
    \label{eq:constraintsWerner2.1} 
    \end{align}
    \endgroup
    We then replace $y^+_{ij}$ in~\eqref{eq:constraintsWerner2.1} by $\sum_{j' \in \epsilon(j)} y_{ij'}$ and denote the respective vectors of the surrogates as $\vec{\sigma}, \vec{\gamma}, \vec{v}, \vec{z}$.
    \item \textit{Objective function}: instead of maximizing the objective function from~\eqref{eq:objDirect}, we take logarithm and minimize its negation.
    Using the convention $\ln(0) \!=\! -\infty$, which acts as a barrier, the reformulated objective becomes
    \begingroup
    \setlength{\abovedisplayskip}{5pt}   % space above equations
    \setlength{\belowdisplayskip}{5pt}   % space below equations
    \setlength{\abovedisplayshortskip}{2pt}
    \setlength{\belowdisplayshortskip}{2pt}
    \begin{align*}
    - \sum_{i \in [k]} \bigg(\ln(\max_{j' \in \delta^{-}(s_i)}\! x_{ij'}) + \ln{f_i}(e^{z_i})\bigg) \,.
    \end{align*}
    \endgroup
    % This means that we need to initialise the optimisation algorithm in the domain $\{\vec{z}: f_i(z_i)>0\, \forall i\}$.
    Now, the term $-\ln(\max_{j' \in \delta^{-}(s_i)}\! x_{ij'})$ can be convexly reformulated using the perspective transformation~\cite{gunluk2010perspective}:
    \begingroup
    \setlength{\abovedisplayskip}{5pt}   % space above equations
    \setlength{\belowdisplayskip}{5pt}   % space below equations
    \setlength{\abovedisplayshortskip}{2pt}
    \setlength{\belowdisplayshortskip}{2pt}
    \begin{align}
    &\lambda_i = \sum_{j' \in \delta^{-}(s_i)} t_{ij'}\, \quad \text{with} \\
    &t_{ij'} =
    \begin{cases}
        -y_{ij'} \ln\!\left(\nicefrac{x_{ij'}}{y_{ij'}}\right) & y_{ij'} > 0 \\
        0 & x_{ij'} = 0,\; y_{ij'} = 0 \\
        \infty & \text{otherwise,}
    \end{cases}
    \end{align}
    \endgroup
    due to the presence of the constraint
    \begingroup
    \setlength{\abovedisplayskip}{5pt}   % space above equations
    \setlength{\belowdisplayskip}{5pt}   % space below equations
    \setlength{\abovedisplayshortskip}{2pt}
    \setlength{\belowdisplayshortskip}{2pt}
    \begin{align*}
    \sum_{j' \in \delta^{-}(s_i)} y_{ij'} = 1\,,\quad y_{ij'} \in \{0,1\}~\forall i\,.
    \end{align*}
    \endgroup
    Here, convexity of $-\ln$ ensures that the perspective reformulation is convex in $(\overset{\leftrightarrow}{x},\vec{y})$~\cite[Sect. 3.2.6]{boyd2004convex}. 
    
    For the fidelity component of the objective function, we introduce the shorthand ${F_i(z)\!:=\!\ln f_i(e^{z}), \, i\!\in\![k]}$.  
    Following~\cite{vardoyan2023quantum}, from now on we only consider three entanglement measures: SKF, a lower bound to DE, and negativity.  
    The function $F_i$ is concave for negativity, whereas for the first two measures, we use the corresponding concave envelope $\hat{F}_i$, which closely approximates $F_i$.  
    For the derivation and justification of this approximation see~\eqref{eq:Fhat} and Fig.~\ref{fig:fHat} in the Appendix.
    We will also empirically validate the approximation accuracy in Sect.~\ref{sec:eval}.
\end{itemize}

% The observations above imply that the single-path entanglement routing problem can be expressed as the following MICP for the specified choices of entanglement measures.

The observations above imply that the single-path entanglement routing problem can be expressed as the MICP~\eqref{eq:modObjective}--\eqref{eq:zMICP} for the specified choices of entanglement measures.

\begin{remark}
% \looseness = -1
The only potential source of inexactness in the formulation is the overestimator $\hat{F}_i$.
For networks with sufficiently high entanglement generation rates per demand, the MICP remains exact because $\hat{F}_i$ differs from $F_i$ only beyond a threshold $\hat{z}$. In high-rate regimes, the optimal solution does not exceed this threshold, as increasing the rate is more advantageous than improving fidelity ($f_i \!\le\! 1$ for considered entanglement measures).
This can be readily verified from the solution. Otherwise, rerunning the MICP with a concave underestimator $\breve{F}_i$ in place of $\hat{F}_i$ provides a bound on the approximation error, given by the gap between the two optimal values.
See Fig.~\ref{fig:fHat} and the appendix for further details.
%
% The only possible source of inexactness in the MICP formulation is the overestimator $\hat{F}_i$.
% For networks that sustain a sufficiently high entanglement generation rate per demand, however, the MICP remains exact.
% This holds because $\hat{F}_i$ diverges from $F_i$ only beyond a threshold $\hat{z}$, and in high-rate regimes the optimal solution never assigns fidelity beyond $\hat{z}$.
% This behavior can be directly confirmed from the obtained solution.
% If exactness is uncertain, a second MICP can be solved using the concave underestimator $\breve{F}_i$; the gap between the two optimal values then bounds the approximation error.
% The estimators $\hat{F}_i$ and $\breve{F}_i$ are shown in Fig.~\ref{fig:fHat}, with derivations provided in the appendix.
\end{remark}

The MICP can be solved using standard MICP solvers, which remains \textit{practical} for instances with moderate size.
For larger networks, we propose an upper bound and a randomized heuristic as follows.

\noindent \textbf{Upper bound}: we obtain a
convex relaxation of the MICP by modifying the integrality constraint~\eqref{eq:yBool} to:
${0 \!\le\! y_{ij'} \!\le\! 1,\, \forall i,j'.}$
We also explicitly include the following constraints, which were automatically true earlier due to~\eqref{eq:yBool}. 
The first constraint is introduced to mitigate path splitting and it allows us to drop either~\eqref{eq:yIncoming} or~\eqref{eq:yOutgoing}, while the second one enforces~\eqref{eq:y+}.
% To mitigate path splitting, we explicitly introduce:
\begingroup
\setlength{\abovedisplayskip}{5pt}   % space above equations
\setlength{\belowdisplayskip}{5pt}   % space below equations
\setlength{\abovedisplayshortskip}{2pt}
\setlength{\belowdisplayshortskip}{2pt}
\begin{align}
    &\sum_{j' \in \delta^{+}(v)}\!\!\! y_{ij'} = \!\! \sum_{j' \in \delta^{-}(v)}\!\!\! y_{ij'}, ~\forall v \!\notin\! \{s_i,t_i\}, \, \forall i \!\in\! [k]\,, \label{eq:flowY} \\
    &\sum_{j' \in \epsilon(j)} y_{ij'}\le 1, ~ \forall j \!\in [l]\,,\,\forall i \!\in\! [k]\,.
\end{align}
\endgroup
% which was automatically true earlier.

% Also, in absence of the original integrality constraints $y_{ij'} \!\in\! \{0,1\}$, we explicitly include~\eqref{eq:y+}:
% \begingroup
% \setlength{\abovedisplayskip}{4pt}   % space above equations
% \setlength{\belowdisplayskip}{2pt}   % space below equations
% \setlength{\abovedisplayshortskip}{0pt}
% \setlength{\belowdisplayshortskip}{0pt}
% \begin{align*}%\label{eq:flowY_}
%      \sum_{j' \in \epsilon(j)} y_{ij'}\le 1, ~ \forall j \!\in [l]\,,\,\forall i \!\in\! [k]\,.
% \end{align*}
% \endgroup

The relaxation provides an upper bound to the maximum achievable utility under single-path routing.

\noindent \textbf{Randomized heuristic}:
Following the idea of randomized rounding~\cite{raghavan1987randomized}, we introduce a heuristic as follows:
\begin{itemize}
    \item \textit{Route sampling}: for each demand, the output variables $y_{ij'}$ from the relaxed MICP are interpreted as probabilities. 
    Employing the path-stripping algorithm~\cite{raghavan1987randomized}, we extract directed paths between each source and destination, with corresponding weights. 
    Because the total outgoing and incoming probabilities at the source and destination, respectively, sum to one (constraints~\eqref{eq:yOutgoingSource} and~\eqref{eq:yIncomingSink}) and flow is conserved at intermediate nodes~\eqref{eq:flowY}, this procedure indeed yields a valid probability distribution over output routes for each demand; see~\cite{raghavan1987randomized} for details.
    \looseness = -1
    \item \textit{Allocation optimization for fixed routes}: once a fixed route is sampled for each demand, we use the convex QNUM formulation~\cite{kar2024convexification} to compute the optimal rate allocation $\vec{x}$.
    % that maximises the overall network utility.
    % For the entanglement measures considered in this work, this allocation subproblem is convex and can therefore be solved efficiently~\cite{kar2024convexification}.
\end{itemize}
% % \vspace{-3pt}

\looseness = -1 
\noindent Next, we propose an alternative heuristic and upper bound inspired by \textit{minimum-congestion routing} in classical networks~\cite{raghavan1987randomized}, both of which often outperform their counterparts on evaluated examples. 
The use of minimum-congestion routing is motivated by the dependence of a route’s e2e Werner parameter on the congestion of its constituent links~\eqref{eq:wernerBidir}.
Both the heuristic and upper bound begin by solving an LP that provides a lower bound on maximum congestion. 
Overall, the heuristic is computationally faster than the previous one, while the upper bound requires solving an MICP with $k$ binary and $8k\!+\!1$ total variables in the second step, substantially fewer than in the original formulation.
% % \vspace{-3pt}

\subsection{Minimum Congestion-based Heuristic and Upper Bound}\label{sec:minCong}
% % \vspace{-1pt}

\begin{definition}[Maximum congestion]
For a valid routing $\tilde{A}$ as introduced in~\eqref{eq:setRoutes}, its maximum congestion is defined as the highest number of routes traversing any single link, i.e., 
\begingroup
\setlength{\abovedisplayskip}{5pt}   % space above equations
\setlength{\belowdisplayskip}{5pt}   % space below equations
\setlength{\abovedisplayshortskip}{2pt}
\setlength{\belowdisplayshortskip}{2pt}
$$c_\textnormal{max}(\tilde{A}) = \max(\tilde{A} \vec{1})\,.$$
\endgroup
\end{definition}
% % \vspace{-3pt}

\looseness=-1
As shown in~\cite[Sect.~3]{raghavan1987randomized}, minimizing the maximum congestion over all possible routings can be formulated as a mixed-integer linear program and then relaxed to an LP for large networks. 
Let $\underline{c}$ denote the optimal congestion value obtained from the LP; then $\lceil \underline{c} \rceil$ provides a lower bound on the maximum congestion. 
By interpreting the fractional LP solution as probabilities,~\cite{raghavan1987randomized} proposes a randomized routing strategy, which forms the basis of our heuristic.

\noindent\textbf{Randomized heuristic:}
For each demand, a route is sampled according to the randomized routing strategy of~\cite[Sect.~3]{raghavan1987randomized}.
Given these sampled routes, the convex QNUM formulation~\cite{kar2024convexification} is applied to compute the optimal rate allocation $\vec{x}$.
This heuristic is computationally faster than the previous one, as the min-congestion LP is faster than the MICP's relaxation.

\noindent \textbf{Upper bound}: first, we introduce the following notation: the $j$th largest element among the link-specific constants ${\vec{d} := (d_1,d_2,\dots, d_l)}$ is denoted by $d_{(j)}$. Given a lower bound to the maximum congestion level $\ceil{\underline{c}}$, we can obtain an upper bound to the achievable network utility by considering the following optimistic scenario: 
% % \vspace{-3pt}
\begin{itemize}
    \looseness = -1
    \item \textit{Single-link congestion:} we assume all but one link support at most one route, while the remaining (congested) link experiences a congestion level of $\lceil \underline{c} \rceil$.
    The corresponding rate–fidelity trade-off constant is taken as $d_{(1)}$.
    \item \textit{Shortest-path routing:} all demands are assumed to be routed through their respective shortest paths, with $\lambda_i$ denoting the corresponding path length for demand $i \!\in\! [k]$. 
    Furthermore, the e2e Werner parameter can be upper bounded by assuming that the rate–fidelity trade-off constants along uncongested links attain their best possible combination: $\{d_{(1)}, d_{(2)}, \dots, d_{(\lambda_i-1)}\}$.
    
\end{itemize}

\looseness = -1 Determining the maximum achievable network utility then reduces to selecting the subset of demands routed through the congested link. 
Let $\alpha_i \!\in\! \{0,1\}$ denote the route assignment variable for demand $i$, and let $\vec{\alpha}$ denote the corresponding vector. 
The upper bound can then be formally expressed as follows. 
% The proof follows directly from the reasoning above and is omitted due to space constraints.

\begin{proposition}
    Given a maximum congestion level $\ceil{\underline{c}}$, the maximum achievable network utility~\eqref{eq:uRouteSingleWd} can be bounded as
    {\small
    \begin{align}
        &\max\limits_{\substack{\vec{0} \preceq \vec{x},\,\Tilde{A} \vec{x} \preceq \vec{d}, \\ \tilde{A} \in \mathcal{P}(A)}} ~ \prod_{i \in [k]} x_i f_i\bigg(\prod\limits_{j=1}^l \bigg( 1\!-\!\frac{\langle \tilde{A}_j,\vec{x} \rangle}{d_j}\bigg)^{\tilde{a}_{ji}}\bigg)  \label{eq:optAssignBound}\\
        \le & \max_{(\vec{x}, \vec{\alpha}) \in S}\! \prod_{i \in [k]}\! x_i f_i\bigg(\!\Big(1\!-\!\!\frac{x_i\!+\! \alpha_i\sum_{i' \ne i}\! \alpha_{i'}x_{i'}}{d_{(1)}}\!\Big) \!\!\!\!\!\prod\limits_{j \in [\lambda_i-1]} \!\!\!\!\Big(\!1\!-\!\frac{x_i}{d_{(j)}}\Big)\!\bigg),\nonumber
    \end{align}}
    where
    \begin{align*}
        S\!:=\!\{(\vec{x}, \vec{\alpha}): & \vec{\alpha}'\vec{1} =\ceil{\underline{c}}, 0 \le x_i \le d_{(\lambda_i)}, \\
        &~x_i+\alpha_i \!\sum\limits_{i' \ne i}\!\! \alpha_{i'}x_{i'} \!\le\! d_{(1)} \forall i\}.
    \end{align*}
    $$$$
\end{proposition}
\begin{proof}
    Let us fix a routing $\tilde{A} \!\in\! \mathcal{P}(A)$ and let $j^*$ denote the index of the link with maximum congestion under $\tilde{A}$.
    Now,
    % We denote the corresponding congestion level by $\tilde{c}$.
    % We will often drop the argument of $j^*$ for brevity. 
    % Under the optimistic scenario described above, the non-decreasing nature of $f_i$'s implies
    \begin{align*}
        & \prod\limits_{j=1}^l \bigg( 1\!-\!\frac{\langle \tilde{A}_j,\vec{x} \rangle}{d_j}\bigg)^{\tilde{a}_{ji}} \\
        = & \bigg( 1\!-\!\frac{\langle \tilde{A}_{j^*},\vec{x} \rangle}{d_j}\bigg)^{\tilde{a}_{j^* i}} \prod\limits_{j \ne j^*} \bigg( 1\!-\!\frac{\langle \tilde{A}_j,\vec{x} \rangle}{d_j}\bigg)^{\tilde{a}_{ji}} \\
        \le & 
        \begin{cases}
            \prod\limits_{j \in [\lambda_i]} \!\!\bigg(1\!-\!\frac{x_i}{d_{(j)}}\bigg),~ \tilde{a}_{j^* i} = 0 \\
            \bigg( 1\!-\!\frac{x_i+  \sum_{i' \ne i} \tilde{a}_{j^* i'} x_{i'}}{d_{(1)}}\bigg) \prod\limits_{j \in [\lambda_i-1]} \!\!\bigg(1\!-\!\frac{x_i}{d_{(j)}}\bigg),~ \tilde{a}_{j^* i} = 1 
        \end{cases}
    \end{align*}
    
    Therefore, the network utility can be bounded as
    \begin{align}
    % \label{eq:minConBd}
        & \prod_{i \in [k]} x_i f_i\bigg(\prod\limits_{j=1}^l \bigg( 1\!-\!\frac{\langle \tilde{A}_j,\vec{x} \rangle}{d_j}\bigg)^{\tilde{a}_{ji}}\bigg) \nonumber\\
        \le & \prod_{i \in [k]} x_i f_i\bigg(\bigg( 1\!-\!\frac{x_i\!+\! \tilde{a}_{j^* i} \sum_{i' \ne i} \tilde{a}_{j^* i'} x_{i'}}{d_{(1)}}\bigg) \!\!\prod\limits_{j \in [\lambda_i-1]} \!\!\bigg(1\!-\!\frac{x_i}{d_{(j)}}\bigg)\bigg) \nonumber
    \end{align}
    
    Now, we maximize the LHS over all possible routings subject to capacity constraints.
    To maximize the RHS, we first observe that for any routing $\tilde{A}$, $\tilde{A}_j \vec{1} \ge \ceil{\underline{c}},\, \forall j$.
    Further, the maximum allocation on route $i$ is bounded by $d_{(\lambda_i)}$ and any allocation must ensure that each $f_i$ has a non-negative argument on the RHS.
    This establishes~\eqref{eq:optAssignBound}.
\end{proof}

\looseness = -1 We formulate the RHS of~\eqref{eq:optAssignBound} as the following optimization problem. 
We again minimize the negative logarithm of the objective, introduce surrogates for the e2e Werner parameters, and approximate $F_i(z) (\!=\! \ln f_i(e^{z}))$ by $\hat{F}_i(z)$.

% % \vspace{-10pt}
\begin{align}
    \min\limits_{\substack{\vec{x},\vec{\alpha},\vec{\eta},S,\vec{\delta} \\ \vec{\zeta},\vec{\gamma},\vec{v},\vec{z} }} &-\sum_{i \in [k]} \ln{x_i} \!-\!\!\sum_{i \in [k]} \!\hat{F}_i\big(z_i\big) \label{eq:optAssign} \\ 
    \text { s.t. } & 0 \leq x_i \leq d_{\left(\lambda_i\right)} ,~ \forall i \in[k] \\ 
    & \alpha_i \in \{0,1\} ,~ \forall i \in[k]\,,\quad \vec{\alpha}'\vec{1}=\ceil{\underline{c}} \\ 
    & \eta_i=\alpha_i x_i ,~ \forall i \in[k] \label{eq:eta}\\ 
    & S=\sum_{i \in [k]} \eta_i \\ 
    & \delta_i=S-\eta_i ,~ \forall i \in[k] \\ 
    & \zeta_i=\alpha_i \delta_i ,~ \forall i \in[k] \label{eq:zeta}\\ 
    & \gamma_i=x_i+\zeta_i ,~ \forall i \in[k] \\ 
    & \gamma_i \le d_{(1)},~ \forall i \in[k] \\
    & v_{i 1}=\ln \left(1-\nicefrac{\gamma_i}{d_{(1)}}\right) ,~ \forall i \in[k] \label{eq:v1}\\
    & v_{i j}=\ln \big(1\!-\!\nicefrac{x_i}{d_{(j-1)}}\big) ,~ \forall j \!\in\![\lambda_i] \backslash\{1\}, \forall i \!\in\! [k] \label{eq:v}\\ 
    & z_i=\sum_{j \in [\lambda_i]} v_{i j}, ~\forall i \in [k] ~\label{eq:zMICPmc}
\end{align}
% % \vspace{-10pt}

\looseness = -1
Since $\alpha_i \!\in\! \{0,1\}$, we obtain an MICP formulation by replacing~\eqref{eq:eta} and~\eqref{eq:zeta} with their exact McCormick relaxations:
\begingroup
\setlength{\abovedisplayskip}{5pt}   % space above equations
\setlength{\belowdisplayskip}{5pt}   % space below equations
\setlength{\abovedisplayshortskip}{2pt}
\setlength{\belowdisplayshortskip}{2pt}
\begin{align*}
    & \eta_i \!\le\! d_{(\lambda_i)}\alpha_i,\, \eta_i \!\ge\! 0,~ \eta_i \!\le\! x_i,\, \eta_i \!\ge\! x_i \!-\!d_{(\lambda_i)}(1\!-\!\alpha_i)~ \forall i \\
    & \zeta_i \!\le\! d_{(1)}\alpha_i,\, \zeta_i \!\ge\! 0,~ \zeta_i \!\le\! \delta_i,~ \zeta_i \!\ge\! \delta_i \!-\!d_{(1)}(1\!-\!\alpha_i)~ \forall i~,
\end{align*}
\endgroup

\noindent \looseness = -1 and modifying~\eqref{eq:v1} and~\eqref{eq:v} by making $v_{ij}$s less than or equal to the respective RHS, which keeps the problem unchanged as $\hat{F}_i$'s are increasing.
In contrast to the original MICP, which involves $2kl$ binary variables, the present formulation requires only $k$ such variables, making it \textit{practical} to solve under wider circumstances.
For instances with a large number of demands, the integrality constraints can be relaxed to obtain an upper bound.

% is a monotonic (increasing) function
% \begin{align*}
% & v_{i 1} \le \ln \left(1-\gamma_i\right) ,~ \forall i \in[k] \\
%     & v_{i j} \le \ln \bigg(1\!-\!\frac{x_i}{d_{(j-1)}}\bigg) ,~ \forall j \!\in\![\lambda_i] \backslash\{1\}, \forall i \!\in\! [k] 
% \end{align*}
% This modification does not alter the problem due to $\hat{F}_i$ is a monotonic (increasing) function.
% Since the current MICP involves significantly fewer optimisation variables than the original formulation, its relaxation tends to yield a tighter upper bound in real-world network examples.

% Observe that
% \begin{align}
%     z_i^{\lambda_i} &= \Big(1\!-\!x_i\Big)^{\lambda_i-\alpha_i}\Big(1\!-\! \!\!\sum_{i' \in [k]}\!\! \alpha_{i'}x_{i'}\Big)^{\alpha_i}~ \forall i \\
%     \implies z_i &\le 1-x_i -\frac{\alpha_i}{\lambda_i} \sum_{i' \ne i}\!\! \alpha_{i'}x_{i'} =:\beta_i\,, \label{eq:amGm}
% \end{align}
% by AM-GM inequality.
% As we are now looking for a lower bound of~\eqref{eq:optAssign} and $F$ is a non-decreasing function, we change the last constraint in~\eqref{eq:optAssign} to $z_i = \beta_i$, which still does not yield a convex region due to the presence of the trilinear terms.
% Thus, we further relax the region to its convex hull via repeated application of McCormick relaxation.
% % \vspace{-5pt}
\section{Numerical Evaluations}\label{sec:eval}
\everypar{\looseness = -1}

In this section, we empirically evaluate the proposed heuristics and upper bounds on real-world networks. 
We consider the BREN ($10$ nodes, $11$ links), UNIC ($15$ nodes, $17$ links) and ARNES ($17$ nodes, $20$ links) topologies from TopologyBench~\cite{matzner2024topology}, a repository of optical fiber networks. 
The network sizes are selected such that the optimal utility computation via the MICP~\eqref{eq:modObjective}--\eqref{eq:zMICP} remains tractable, which we use for benchmarking the heuristics and upper bounds. 
The number of demands~$k$ is chosen proportional to the network size~$n$, specifically ${k \!\in\! \{4,6,8\}}$ for BREN and ${k \!\in\! \{6,8,10\}}$ for others. 
Further, demands (SD pairs) are added \textit{incrementally} according to the table below.

\setlength{\tabcolsep}{2pt} 
\begin{table}[h!!]
\centering
% \caption{Parameter definitions.}
\begin{tabular}{@{} l l @{}}
\hline
\textbf{Network} & \textbf{Demands (incremental)} \\ 
\hline
BREN & $$\{(6,10),(10,4),(3,10),(9,1)\},\,\{(5,1),(8,10)\},\,\{(7,2),(3,7)\}$$ \\
UNIC & $$\{(8,11),(5,13),(9,1),(2,8),(11,1),(10,2)\},\,\{(7,3),(11,14)\},\,\{(15,14),(1,9)\}$$ \\
ARNES & $$\{(2,3),(10,5),(16,2),(12,7),(3,4),(17,2)\},\,\{(15,12),(3,11)\},\,\{(3,15),(2,4)\}$$ \\
\hline
\end{tabular}
\end{table}

\noindent We refer the reader to~\cite{matzner2024topology} for detailed network information, including link lengths, node locations, and name–ID mapping.  
The topology is shown in Fig.~\ref{fig:topologies}.
Also, we show the evaluation for SKF and negativity in Fig.~\ref{fig:whole} and for DE in Fig.~\ref{fig:DE}.

\begin{figure*}[t]
\centering
\begin{subfigure}{0.31\textwidth}
    \centering
    \includegraphics[width=1\textwidth]{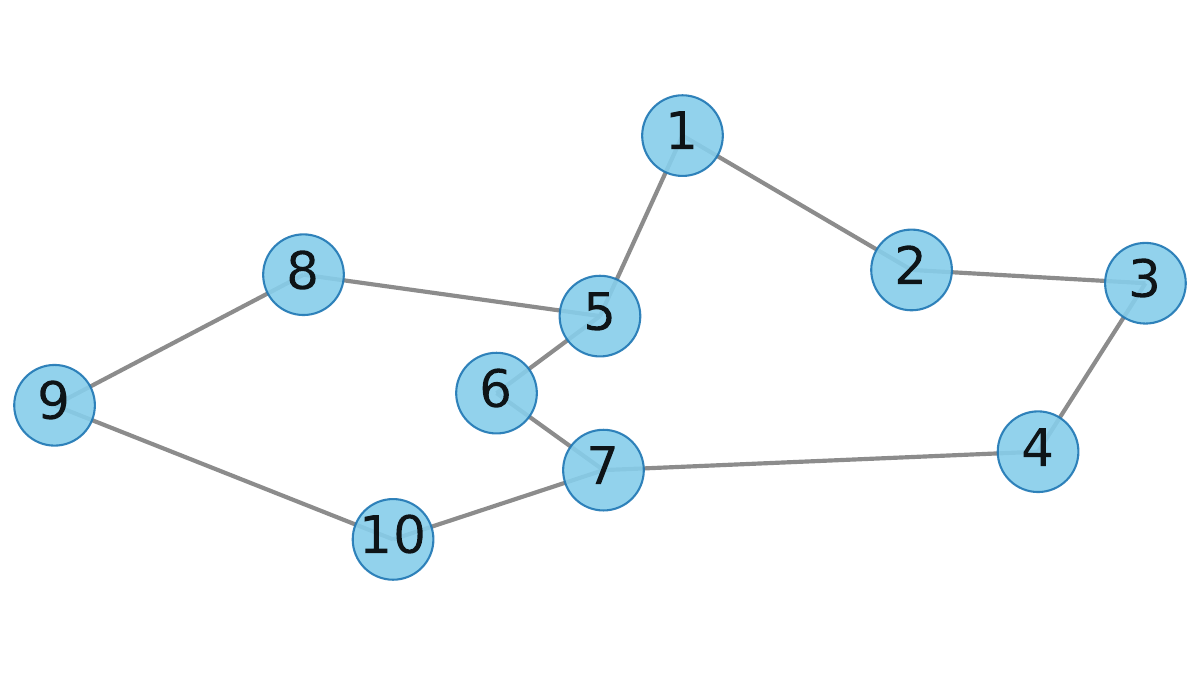}
    \caption{\label{fig:bren}%
    BREN}
\end{subfigure}
    \hfill
\begin{subfigure}{0.31\textwidth}
    \centering
    \includegraphics[width=1\textwidth]{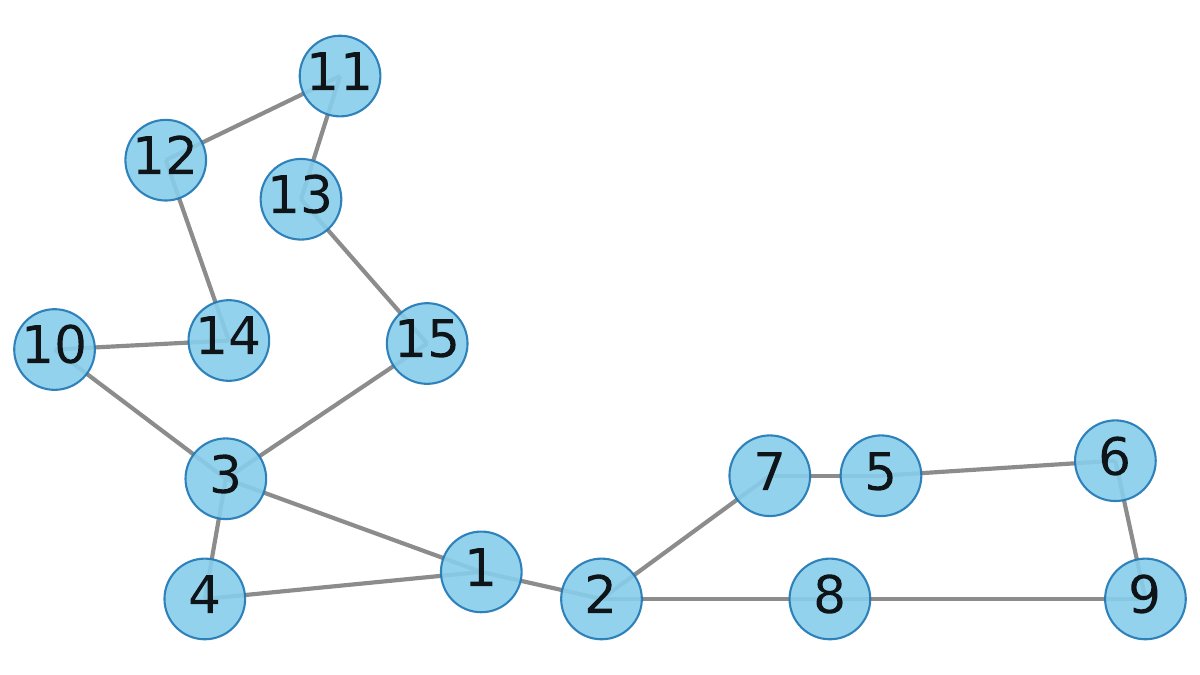}
    \caption{\label{fig:unic}%
    UNIC}
\end{subfigure}
\hfill
\begin{subfigure}{0.31\textwidth}
    \centering
    \includegraphics[width=1\textwidth]{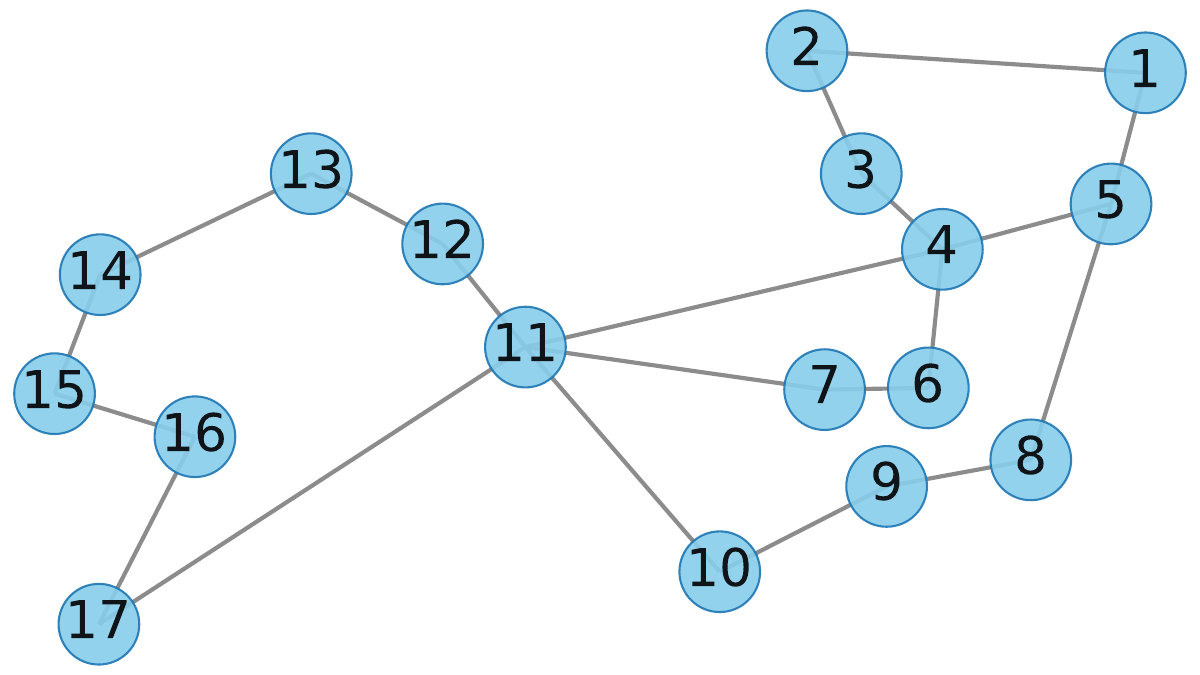}
    \caption{\label{fig:arnes}%
    ARNES}
\end{subfigure}%
\caption{\label{fig:topologies}%
Topologies of the real-world core optical networks~\cite{matzner2024topology} used for evaluations.}%
\vspace{-0pt}
\end{figure*}

\begin{figure*}[t]
    \centering
    \includegraphics[width=\textwidth]{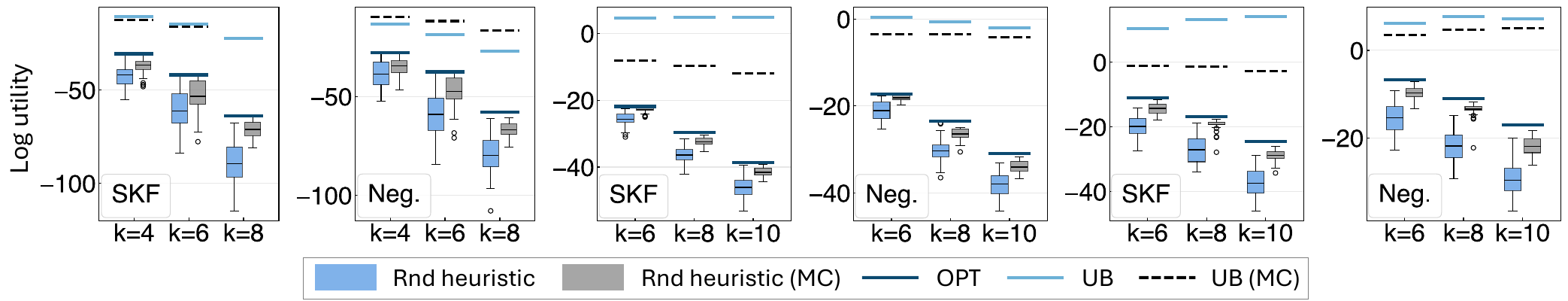}
\caption{\label{fig:whole}%
Performance of the randomized heuristics (Rnd heuristic) and upper bounds (UB) vis-à-vis the optimum log utility (OPT) calculated via MICP~\eqref{eq:modObjective}--\eqref{eq:zMICP}, min-congestion-based metrics are marked as (MC). Results are shown for entanglement measures SKF and negativity (Neg.) on BREN ($10$ nodes, left $2$ subplots), UNIC ($15$ nodes) and ARNES ($17$ nodes, right $2$ subplots) networks~\cite{matzner2024topology} for varying demand counts $k$. The MC heuristic outperforms its counterpart \textit{on average}, and the MC upper bound is often closer to OPT.
Recall that log utility here equals utility as per NUM~\cite{kelly1997charging,vardoyan2023quantum}.}%
\vspace{-0pt}
\end{figure*}

% % \vspace{-5pt}
% \noindent Due to space constraints, we refer the reader to~\cite{matzner2024topology} for detailed network information, including topology, link lengths, node locations and name–ID mapping.
% Also, here we provide the evaluation only for two entanglement measures: SKF and negativity.

\looseness = -1 Given the link lengths $L_j$~\cite{matzner2024topology}, the link-level constants $d_j$ are computed from~\eqref{eq:genRate} using $\eta_j \!=\! 10^{-0.02 L_j}$ and by setting $\kappa_j \!=\! 0.1$, $T_j \!=\! 10^{-3}$s for all links, which represents the current state of hardware efficiency~\cite{vardoyan2023quantum,kar2024convexification}. 
Using the constants $d_j$ and network adjacency structure, we solve the MICP~\eqref{eq:modObjective}--\eqref{eq:zMICP} to obtain the optimal utility; its relaxation provides the first upper bound. 
The min-congestion-based upper bound is derived by solving the second MICP~\eqref{eq:optAssign}--\eqref{eq:zMICPmc} with $k$ binary variables. 
We use the MOSEK solver with CVXPY~\cite{diamond2016cvxpy} for our computations.
Since CVXPY only allows functions adhering to its disciplined convex programming syntax, we use piecewise linear approximation of $\hat{F}_i$~\eqref{eq:modObjective} with high granularity.

\looseness = -1 Both upper bounds and the corresponding randomized heuristics are compared against the optimum for SKF and negativity in Fig.~\ref{fig:whole} and for DE in Fig.~\ref{fig:DE}. 
We observe that the min-congestion-based upper bound often outperforms its counterpart, and its heuristic also achieves higher \textit{average} performance. 

\noindent \textbf{Approximation error}:
While the first MICP~\eqref{eq:modObjective}--\eqref{eq:zMICP} computes the maximum network utility exactly for negativity, we calculate the approximation error for SKF and DE by rerunning it with the concave underestimator $\Breve{F}_i$ in place of $\hat{F}_i$ in~\eqref{eq:modObjective}. 
The approximation error is not shown in Fig.~\ref{fig:whole}~and~\ref{fig:DE} as the highest observed relative error was $0.00051\%$ (ARNES, SKF, $10$ demands), empirically supporting our claim of accuracy of the formulation.
Finally, recall that the logarithmic utility in Fig.~\ref{fig:whole}~and~\ref{fig:DE} is equivalent to the utility in NUM parlance~\cite{kelly1997charging,vardoyan2023quantum}; we adopted this form~\cite{johansson1991introduction} only to motivate our formulation in Sect.~\ref{sec:assumptions}.

\begin{figure}[t]
    \centering
    \includegraphics[width=\columnwidth]{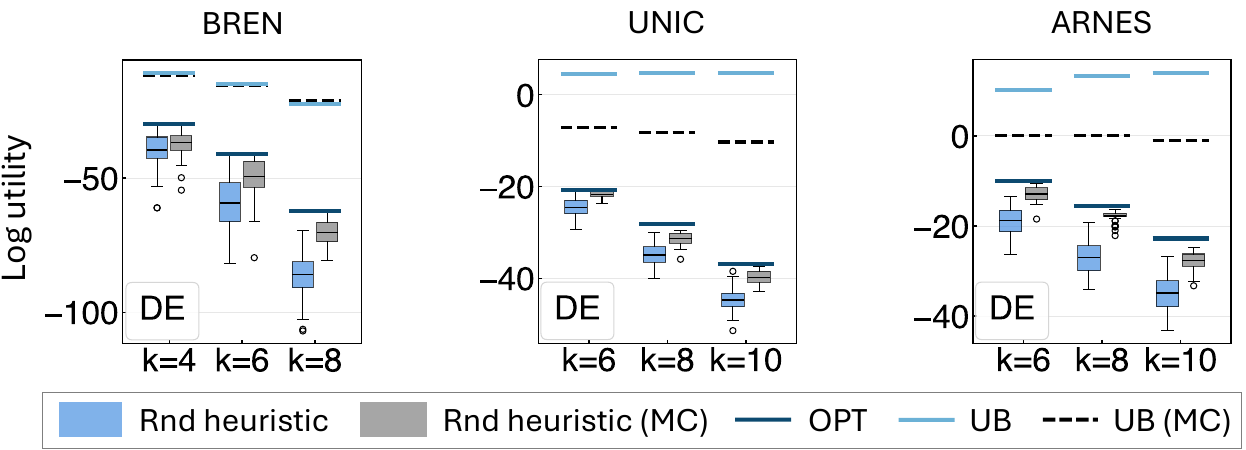}
\caption{\label{fig:DE}%
Performance of the randomized heuristics with the lower bound to DE~\eqref{eq:de} as the entanglement measure, the behavior is qualitatively similar to that of SKF in Fig~\ref{fig:whole}.}%
\vspace{0pt}
\end{figure}

\section{Conclusion}\label{sec:conclusion}
% % \vspace{-4pt}
% In this work, we considered the single-path version of the utility-based entanglement routing problem where we aimed to find an optimal routing for maximising the overall network utility.
% We specifically used SKF, a lower bound to DE and negativity as entanglement measures to quantify the network utility.
% We formulated the problem as an MICP, which was exact for negativity or when the network supported sufficiently high entanglement generation rates.
% For other cases, the MICP provided a \textit{close} approximation and we proposed a method to bound the approximation error.
% Further, for large networks, we proposed a randomised rounding-based heuristic and an associated upper bound by relaxing the MICP.
% We then introduced a min-congestion-based randomised heuristic and upper bound, which were faster to compute.
% Finally, we evaluated the heuristics and upper bounds vis-à-vis the actual optimum in real-world networks, where the latter approach often outperformed the former.
% Our framework can be used for extending the classical flow-based and QoS-aware routing ideas to quantum networks.
\looseness = -1 In this work, we considered the single-path version of the utility-based entanglement routing problem, aiming to determine optimal routes that maximize overall network utility.
Network utility was quantified using SKF, a lower bound to DE and negativity as entanglement measures. 
We formulated the problem as an MICP, which provides exact solutions when negativity is used or when the network supports sufficiently high entanglement generation rates.
In other cases, the formulation yields a reasonably close approximation, as seen in real-world examples. 
To ensure scalability in large networks, we proposed a randomized rounding-based heuristic and an upper bound derived from the relaxed MICP. 
We further proposed a computationally faster randomized heuristic and an upper bound based on min-congestion routing, which often outperformed their counterparts on real-world network topologies.
Our framework can be used for extending classical flow-based and QoS-aware routing principles to quantum networks, enabling fair and efficient allocation of quantum communication resources.

% % \vspace{-5pt}
\section*{Acknowledgement}
SK acknowledges support from NWO Vici Program under grant VI.C.222.029 awarded to Stephanie Wehner and thanks her for critical feedback on the manuscript.
{
}
% \vspace{-8pt}
\appendix
% \subsection{Approximating Entanglement Measures}
% % \vspace{-10pt}

% \looseness = -1 Due to space limitations, we refer the reader to~\cite[Eq.~16-18]{kar2024convexification} for definitions of the considered entanglement measures: SKF, a lower bound to DE (as adopted in~\cite{vardoyan2023quantum,kar2024convexification}) and negativity in terms of the Werner parameter.

We first provide definitions for the considered entanglement measures. 
The secret key fraction~\cite{shor2000simple} has the following form for Werner states with Werner parameter $w$:
\begin{align}\label{eq:sk}
    f_\text{sk}(w) \!=\! \max \Big(0,1\!+\!(1\!+\!w) \log_2 \frac{1\!+\!w}{2}+ 
    (1\!-\!w) \log_2 \frac{1\!-\!w}{2} \Big).
\end{align}

Further, following~\cite{vardoyan2023quantum}, we consider a lower bound to distillable entanglement.
For a Werner state with Werner parameter $w$, the lower bound can be expressed as
\begin{align}\label{eq:de}\vspace{-5pt}
    f_\text{de}(w) &\!=\! \max \Big(0,1+\frac{1+3w}{4}\log_2\Big(\frac{1+3w}{4}\Big)+ \nonumber \\
    &\frac{3(1-w)}{4}\log_2\Big(\frac{1-w}{4}\Big)\Big)~.
\end{align}
Finally, a Werner state with Werner parameter $w$, negativity can be expressed as
\begin{align}\label{eq:neg}\vspace{-8pt}
    f_\text{neg}(w) \!=\! \max\Big(0,\frac{3 w-1}{4}\Big)~.
\end{align}
% As mentioned earlier, we focus on three specific entanglement measures ($f_i$).
% Their definitions are provided below.
% \noindent \textit{Secret key fraction (SKF)}:
% %\noindent {\bfseries\small Secret key fraction (SKF)}: 
% the secret key fraction~\cite{shor2000simple} has the following form for Werner states with Werner parameter $w$:
% {\small
% \begin{align}\label{eq:sk}
%     f_\text{sk}(w) \!=\! \max \!\Big(0,1\!+\!(1\!+\!w) \log_2\! \frac{1\!+\!w}{2}\!+\! 
%     (1\!-\!w) \log_2\! \frac{1\!-\!w}{2} \Big).
% \end{align}
% }
% \noindent \textit{Distillable entanglement (DE)}: following previous works on QNUM~\cite{vardoyan2023quantum,kar2024convexification}, we actually consider a lower bound to distillable entanglement.
% For a Werner state with Werner parameter $w$, the lower bound is given by
% {\small
% \begin{align*}
%     f_\text{de}(w) \!=\! \max \Big(0,1\!+\!\frac{1\!+\!3w}{4}\log_2\!\Big(\frac{1\!+\!3w}{4}\Big)\!+\! 
%     \frac{3(1\!-\!w)}{4}\log_2\!\Big(\frac{1\!-\!w}{4}\Big)\Big)~.
% \end{align*}}
%
% \noindent \textit{Negativity}: for a Werner state with Werner parameter $w$, negativity is given by
% \begin{align}\label{eq:neg}% \vspace{-8pt}
%     f_\text{neg}(w) \!=\! \max\Big(0,\frac{3 w-1}{4}\Big)~.
% \end{align}

\looseness = -1 Recall that for each demand $i$, we transform the entanglement measure $f_i \!\in\! \{f_\text{sk}, f_\text{de}, f_\text{neg}\}$ to ${F_i(z) \!=\! \ln(f_i(e^z)), z \!\in\! (z_\text{min}^{(i)},0]}$.
As $F_i$ is concave for negativity but not for the other two measures, we look for a 

\begin{center}
\begin{tabular}{p{0.1\columnwidth}p{0.8\columnwidth}}
  \multicolumn{2}{c}{\textit{General}}\\[1pt]\hline
  $[\iota]$ & $\{1,2,\dots,\iota\}$ for $\iota \in \mathbb{N}$ \\
  $B_p$ & $p$-th row of a matrix $B$ \\
  $B_{:q}$ & $q$-th column of a matrix $B$ \\
  [5pt]
  \multicolumn{2}{c}{\textit{Links}}\\[1pt]\hline
  $l$ & number of links in the network \\
  $d_j$ & a positive constant characterising the physical attributes of the $j$-th link,
  see~\eqref{eq:genRate} \\
  $w_j$ & Werner parameter of a generated pair in link $j$ \\
  $\Vec{w}$ & $(w_1,w_2,\dots,w_l)$ \\
  $\mu_j$ & corresponding entanglement generation capacity of the $j$-th link, $\mu_j:= d_j(1-w_j)$ \\[5pt]
  \multicolumn{2}{c}{\textit{Demands}}\\[1pt]\hline
  \!$(s_i,t_i)$ & $i$th sender-receiver (SD) pair, i.e., demand\\
  $k$ & number of demands \\
  $p_i$ & number of routes serving demand $i$\\
  % $P_i$ & $\sum_{m \in [i]} p_m$\\
  $I_i$ & \looseness = -1 set of simple path (route) indices serving demand $i$ \\[5pt]
  % the $i$-th SD pair, $I_i := [P_i]\!\setminus\! [P_{i-1}]~([P_0]:=\emptyset)$ \\[5pt]
  % \multicolumn{2}{c}{\textit{Routes}}\\[1pt]\hline
  % $r$ & $P_k$ \\
  % $x_m$ & the rate allocated to the $m$-th route\\
  % $y_m$ & $\ln(x_m)$ \\[5pt]
  \multicolumn{2}{c}{\textit{Network and allocations}}\\[1pt]\hline
  $V$ & set of nodes\\
  % $V_0$ & set of terminal nodes, i.e., ${V_0 = \bigcup\limits_{i\in [k]}(\{s_i\}\cup\{t_i\})}$ \\
  $n$ & number of nodes in the network ($|V|$) \\
  % $a_{jm}$ & the binary variable taking value $1$ iff the $m$-th (simple) route passes through the $j$-th link\\
  $A$ & overall link-route incidence matrix encoding all simple paths corresponding to all demands \\
  $\mathcal{P}(A)$ & set of valid single-path routings of $k$ demands\\
  $\Tilde{A}$ & a link-route incidence matrix in $\mathcal{P}(A)$, ${A \!\in\! \{0,1\}^{l \times k}}$, $\Tilde{A}_{:i}\!=\!A_{:m}$ for some $m \!\in\! I_i$, $((\tilde{A}))_{ji} \!=\! \tilde{a}_{ji}$ \\
  % $A \in \{0,1\}^{l \times r}$ \\
  % $A_j$ & the $j$-th row of $A$, $A_j:=(a_{j1},a_{j2},\dots,a_{jr})$\\
  $u_i$ & end-to-end Werner parameter of demand $i$ given routing $\tilde{A}$, ${u_i := \prod\nolimits_{j \in [l]}w_j^{\tilde{a}_{ji}}}$ \\ 
  % $v_i(\vec{y})$ & $\prod\limits_{j=1}^l \big(1-\langle A_j,e^{\vec{y}} \rangle/d_j\big)^{a_{jm}}$, $\Vec{y} \!\in\! \mathbb{R}^r$\\
  $f_i$ & entanglement measure for demand $i$ \\
  $F_i(z)$ & $\!\ln f_i(e^{z})$ \\
  $\hat{F}_i$ & concave envelope of $F_i$ \\
  $x_i$ & rate allocation for demand $i$ \\
  $\Vec{x}$ & $(x_1,x_2,\dots,x_r)$ \\
  $x_{ij'}$ & rate allocated to demand $i$ on $j'$th directed link \\
  $\overset{\leftrightarrow}{x}$ & $(x_{11},x_{12},\dots,x_{1\, 2l},\dots,x_{k1},x_{k2},\dots,x_{k\, 2l})$ \\
  $y_{ij'}$ & $\mathbbm{1}_{x_{ij'}>0}$\\
  $\epsilon(j)$ & the set comprising two directed link indices corresponding to the undirected link $j$ \\
  $\delta^+(v)$ & index set of incoming links at $v$\\
  $\delta^-(v)$ & index set of outgoing links from $v$\\
  % $f_i\!:\![0,1] \!\to\! [0,\beta_i]$, non-decreasing and twice differentiable on $\{z\!:\!f_i(z)\!>\!0\}\!\setminus\!\{1\}$, and $f_i(0) \!=\! 0$ \\
  % $c^{(k)}$ & $\sup \{z: f_i(z)=0\}$\\
  % $T$ &$\{\Vec{y} \!\in\! \mathbb{R}^r\!\!:\langle A_j,e^{\vec{y}} \rangle \!<\!d_j ~ \forall j\}$\\
  % $S_i$ & $\{\Vec{y} \!\in\! T\!:v_i(\Vec{y})\!>\!c^{(i)}\!\}$\\
  % $F_i$ & $\ln f_i$,~ $F_i:(c^{(k)},1]\to \mathbb{R}$\\
  % $c_1^{(k)}$ & unique inflection point of $F_i$, $F_i$ is concave in $(c^{(k)},c_1^{(k)}]$ and convex in $(c_1^{(k)},1)$\\
  \hline
\end{tabular}
\captionof{table}{List of notations: the (undirected) link index $j \!\in\! [l]$, directed link index $j' \in [2l]$, and demand index $i \in [k]$.}
\vspace{2pt}
\label{tab:defs}
\end{center}

\begin{figure}[t]
\centering
\includegraphics[width=0.9\columnwidth]{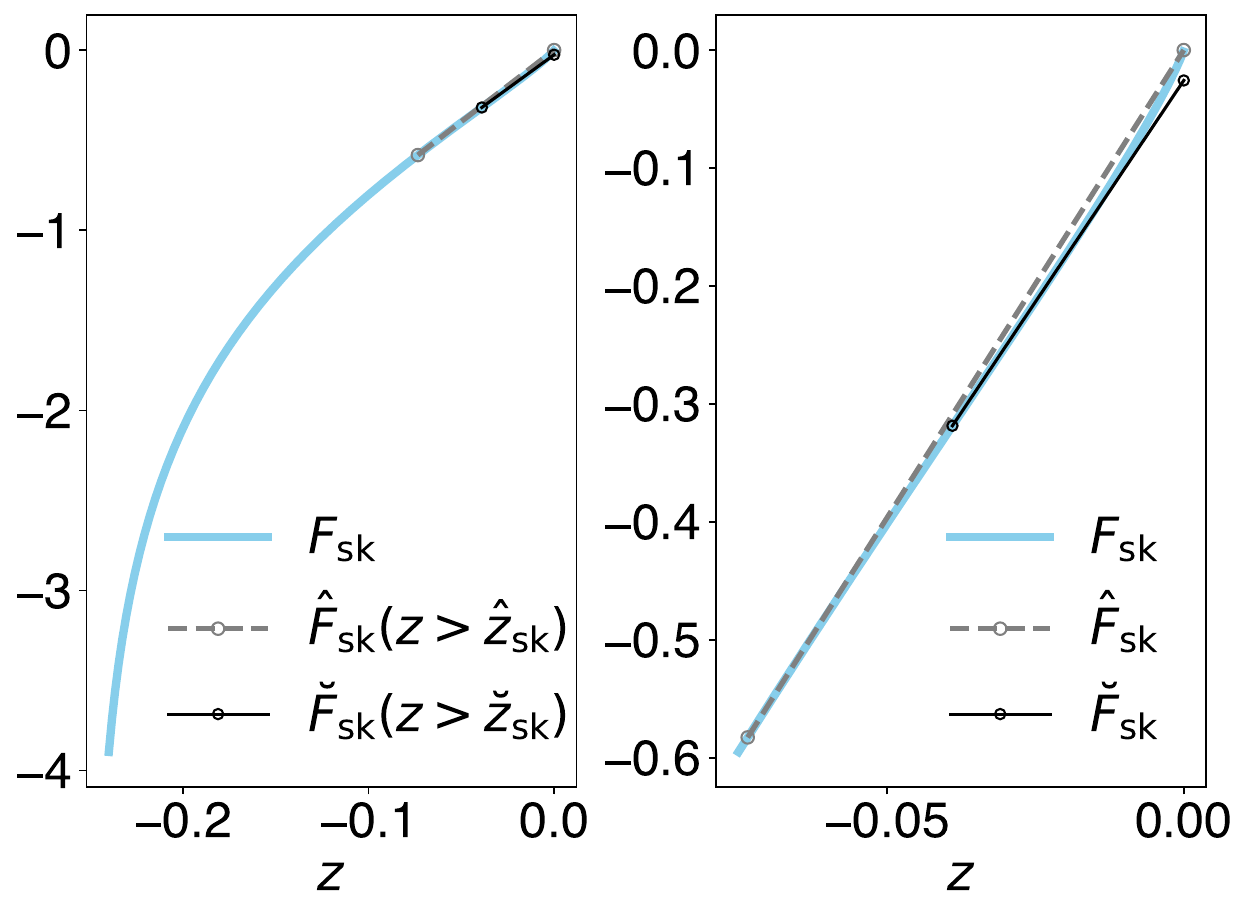}
\caption{\label{fig:fHat}%
Left subplot: estimators of $F_i$ when $f_i = f_\text{sk}$~\eqref{eq:sk}. The right subplot zooms in on the domain of non-concavity of $F_i$ to show the deviation of the estimates. Error bounds:
$|\!|\hat{F}_\text{sk}\!-\!F_\text{sk}|\!|\!<\!0.0165$, $|\!|{F}_\text{sk}\!-\!\breve{F}_\text{sk}|\!|\!<\!0.0258$, 
}%
\vspace{-5pt}
\end{figure}

\noindent concave overestimator $\hat{F}_i$.
The key property that we use for these two measures is that $F_i$ has a unique inflection point $\breve{z}_{(i)}$ and $F_i$ is concave in $(z_\text{min}^{(i)},\breve{z}_{(i)}]$.

\noindent \textbf{Concave enevlope} $\hat{F}_i$:
% Looking at the graph of $F_i$ for SKF and DE, 
We first derive the smallest linear upper bound to $F_i$ passing through $(0,F_i(0))$.
This can be found by calculating the minimal solution of ${F_i'(z) \!=\! F_i(z)/z}$, as $F_i(0)=0$.
% To that end, we find the minimal solution to:
% \begin{align}
%     F_i'(z) = \frac{F_i(z)-F_i(0)}{z-0} = \frac{F_i(z)}{z}\,.
% \end{align}
% The above equation is satisfied by at most two $z$'s as $F_i$ has a unique inflection point~\cite{kar2024convexification}.
The minimal solution $\hat{z}_{(i)}$ is found via the Newton-Raphson method by initializing near $z_\text{min}^{(i)}$.
% sufficiently small point satisfying $F_i(z)>-\infty$.
We then define
\begingroup
\setlength{\abovedisplayskip}{5pt}   % space above equations
\setlength{\belowdisplayskip}{5pt}   % space below equations
\setlength{\abovedisplayshortskip}{2pt}
\setlength{\belowdisplayshortskip}{2pt}
\begin{align}\label{eq:Fhat}
    \hat{F_i}(z) = 
    \begin{cases}
        & F_i(z)\,,~ z_\text{min}^{(i)} <z \le \hat{z}_{(i)} \\
        & F_i(\hat{z}_{(i)})+F_i'(\hat{z}_{(i)}) (z\!-\!\hat{z}_{(i)}) \,, ~ \hat{z}_{(i)}\!<\! z \!\le\! 0 
    \end{cases}
\end{align}
\endgroup
Now, considering the epigraph of $-F_i$ and drawing its convex hull, we can see that $\hat{F_i}$ is indeed the concave envelope of $F_i$.

\looseness = -1 If the $z$ values in the solution of the MICP~\eqref{eq:modObjective}--\eqref{eq:zMICP} are below $\hat{z}_{(i)}$, the MICP is exact for the entanglement routing problem.
Otherwise, we rerun the MICP~\eqref{eq:modObjective}--\eqref{eq:zMICP} with the underestimator $\breve{F}_i$ to bound the approximation error.

\noindent \textbf{Underestimator} $\breve{F}_i$: We simply use the tangent of $F_i$ at the inflection point $\breve{z}_{(i)}$ to define $\breve{F}_i$ beyond $\breve{z}_{(i)}$:
\begingroup
\setlength{\abovedisplayskip}{5pt}   % space above equations
\setlength{\belowdisplayskip}{5pt}   % space below equations
\setlength{\abovedisplayshortskip}{2pt}
\setlength{\belowdisplayshortskip}{2pt}
\begin{align}\label{eq:Fcheck}
    \breve{F_i}(z) = 
    \begin{cases}
        & F_i(z)\,,~ z_\text{min}^{(i)} <z \le \breve{z}_{(i)} \\
        & F_i(\breve{z}_{(i)})+F_i'(\breve{z}_{(i)}) (z\!-\!\breve{z}_{(i)}) \,, ~ \breve{z}_{(i)}\!<\! z \!\le\! 0 
    \end{cases}
\end{align}
\endgroup

We show the accuracy of estimation by plotting the linear parts of $\hat{F}_i$ and $\breve{F}_i$ vis-à-vis $F_i$ in Fig.~\ref{fig:fHat} for SKF.
The plots for DE are qualitatively similar with $|\!|\hat{F}_\text{de}\!-\!F_\text{de}|\!|\!<\!0.0135$, $|\!|{F}_\text{de}\!-\!\breve{F}_\text{de}|\!|\!<\!0.0212$.

\begin{thebibliography}{00}

\bibitem{bennett1984quantum} C. H. Bennett and G. Brassard, ``Quantum cryptography: Public key distribution and coin tossing'', in \textit{Proceedings of the IEEE International Conference on Computer Systems and Signal Processing}, 1984.
% pp. 175--179,

% \bibitem{ekert1992quantum} A. K. Ekert, ``Quantum cryptography and Bell's theorem'', in \textit{Quantum Measurements in Optics}. Springer, 1992, pp. 413--418.

\bibitem{giovannetti2004quantum} V. Giovannetti, S. Lloyd, and L. Maccone, ``Quantum-enhanced measurements: beating the standard quantum limit'', in \textit{Science}, vol. 306, no. 5700, pp. 1330--1336, 2004.

% \bibitem{jozsa2000quantum} R. Jozsa, D. S. Abrams, J. P. Dowling, and C. P. Williams, ``Quantum clock synchronization based on shared prior entanglement'', in \textit{Physical Review Letters}, vol. 85, no. 9, p. 2010, 2000.

\bibitem{broadbent2009universal} A. Broadbent, J. Fitzsimons, and E. Kashefi, ``Universal blind quantum computation'', in \textit{2009 50th Annual IEEE Symposium on Foundations of Computer Science}. IEEE, 2009, pp. 517--526.

\bibitem{kelly1997charging} F. Kelly, ``Charging and rate control for elastic traffic'', in \textit{European Transactions on Telecommunications}, vol. 8, no. 1, pp. 33--37, 1997.
% 

\bibitem{kelly1998rate} F. P. Kelly, A. K. Maulloo, and D. K. H. Tan, ``Rate control for communication networks: shadow prices, proportional fairness and stability'', in \textit{Journal of the Operational Research Society}, vol. 49, no. 3, 1998.
%  pp. 237--252,

\bibitem{vardoyan2023quantum} G. Vardoyan and S. Wehner, ``Quantum network utility maximization'', in \textit{2023 IEEE International Conference on Quantum Computing and Engineering (QCE)}, vol. 1, pp. 1238--1248, 2023.

\bibitem{kar2024convexification} S. Kar and S. Wehner, ``Convexification of the Quantum Network Utility Maximisation Problem'', in \textit{IEEE Transactions on Quantum Engineering}, 2024.

\bibitem{apostolopoulos1998quality} G. Apostolopoulos, R. Guerin, S. Kamat, and S. K. Tripathi, ``Quality of service based routing: A performance perspective,'' \textit{Proceedings of the ACM SIGCOMM}, 1998.
% pp. 17--28,

\bibitem{lin2003multi} X. Lin and N. B. Shroff, ``The multi-path utility maximization problem,'' in \textit{Proceedings of the Annual Allerton Conference on Communication, Control, and Computing}, vol. 41, no. 2, pp. 789--798, 2003.

\bibitem{wu2008utility} J. Wu, M. Lu, and F. Li, ``Utility-based opportunistic routing in multi-hop wireless networks,'' in \textit{Proceedings of the 28th International Conference on Distributed Computing Systems (ICDCS)}, pp. 470--477, 2008.

\bibitem{ahuja1993network} R. K. Ahuja, T. L. Magnanti, and J. B. Orlin, \textit{Network Flows: Theory, Algorithms, and Applications}. Prentice Hall, 1993.

\bibitem{charikar2018multi} M. Charikar, Y. Naamad, J. Rexford, and X. K. Zou, ``Multi-commodity flow with in-network processing,'' \textit{International Symposium on Algorithmic Aspects of Cloud Computing}, pp. 73--101, 2018.

\bibitem{lin2006utility} X. Lin and N. B. Shroff, ``Utility maximization for communication networks with multipath routing,'' \textit{IEEE Transactions on Automatic Control}, vol. 51, no. 5, pp. 766--781, 2006.

\bibitem{abane2025entanglement} A. Abane, M. Cubeddu, V. S. Mai, and A. Battou, ``Entanglement routing in quantum networks: A comprehensive survey,'' in \textit{IEEE Transactions on Quantum Engineering}, 2025.

\bibitem{van2013path} R. Van Meter, T. Satoh, T. D. Ladd, W. J. Munro, and K. Nemoto, ``Path selection for quantum repeater networks,'' in \textit{Networking Science}, vol. 3, no. 1, pp. 82--95, 2013.

\bibitem{gyongyosi2017entanglement} L. Gyongyosi and S. Imre, ``Entanglement-gradient routing for quantum networks,'' in \textit{Scientific Reports}, vol. 7, no. 1, p. 14255, 2017.

\bibitem{pant2019routing} M. Pant, H. Krovi, D. Towsley, L. Tassiulas, L. Jiang, P. Basu, D. Englund, and S. Guha, ``Routing entanglement in the quantum internet,'' in \textit{npj Quantum Information}, vol. 5, no. 1, p. 25, 2019.

\bibitem{chakraborty2020entanglement} K. Chakraborty, D. Elkouss, B. Rijsman, and S. Wehner, ``Entanglement distribution in a quantum network: A multicommodity flow-based approach,'' in \textit{IEEE Transactions on Quantum Engineering}, vol. 1, 2020.
% pp. 1--21,

\bibitem{li2022fidelity} J. Li, M. Wang, K. Xue, R. Li, N. Yu, Q. Sun, and J. Lu, ``Fidelity-guaranteed entanglement routing in quantum networks,'' in \textit{IEEE Transactions on Communications}, vol. 70, no. 10, pp. 6748--6763, 2022.

\bibitem{pouryousef2023quantum} S. Pouryousef, N. K. Panigrahy, and D. Towsley, ``A quantum overlay network for efficient entanglement distribution,'' in \textit{Proceedings of IEEE INFOCOM– IEEE Conference on Computer Communications}, 2023.
% pp. 1--10

\bibitem{le2022dqra} L. Le and T. N. Nguyen, ``DQRA: Deep quantum routing agent for entanglement routing in quantum networks,'' in \textit{IEEE Transactions on Quantum Engineering}, vol. 3, pp. 1--12, 2022.

\bibitem{plenio2005introduction} M. B. Plenio and S. Virmani, ``An introduction to entanglement measures'', in \textit{Quantum Inf. Comput.}, vol. 7, no. 1, pp. 1--51, 2007.

\bibitem{raghavan1987randomized} P. Raghavan and C. D. Tompson, ``Randomized rounding: a technique for provably good algorithms and algorithmic proofs'', in \textit{Combinatorica}, vol. 7, no. 4, pp. 365--374, 1987.
% \bibitem{fitzsimons2017unconditionally} J. F. Fitzsimons and E. Kashefi, ``Unconditionally verifiable blind quantum computation'', in \textit{Physical Review A}, vol. 96, no. 1, p. 012303, 2017.

\bibitem{cabrillo1999creation} C. Cabrillo, J. I. Cirac, P. Garcia-Fernandez, and P. Zoller, ``Creation of entangled states of distant atoms by interference'', in \textit{Physical Review A}, vol. 59, no. 2, pp. 1025, 1999.

\bibitem{dur2007entanglement} W. D{\"u}r and H. J. Briegel, ``Entanglement purification and quantum error correction'', in \textit{Reports on Progress in Physics}, vol. 70, no. 8, 2007.
% pp. 1381,

\bibitem{bennett1996mixed} C. H. Bennett, D. P. DiVincenzo, J. A. Smolin, and W. K. Wootters, ``Mixed-state entanglement and quantum error correction,'' in \textit{Physical Review A}, vol. 54, no. 5, pp. 3824, 1996.

\bibitem{pan1998experimental} J.-W. Pan, D. Bouwmeester, H. Weinfurter, and A. Zeilinger, ``Experimental entanglement swapping: entangling photons that never interacted'', in \textit{Physical Review Letters}, vol. 80, no. 18, p. 3891, 1998.

\bibitem{munro2015inside} W. J. Munro, K. Azuma, K. Tamaki, and K. Nemoto, ``Inside quantum repeaters'', in \textit{IEEE Journal of Selected Topics in Quantum Electronics}, vol. 21, no. 3, pp. 78--90, 2015.

\bibitem{johansson1991introduction} P. O. Johansson, ``An Introduction to Modern Welfare Economics'', \textit{Cambridge University Press}, 1991.

\bibitem{guruswami1999near} 
V. Guruswami, S. Khanna, R. Rajaraman et. al., 
``Near-optimal hardness results and approximation algorithms for edge-disjoint paths and related problems,'' 
\textit{Proceedings of the STOC}, 
pp. 19--28, 1999.


\bibitem{mccormick1976computability} G. P. McCormick, ``Computability of global solutions to factorable nonconvex programs: Part I—Convex underestimating problems,'' Mathematical Programming, vol. 10, no. 1, pp. 147--175, 1976.

\bibitem{gunluk2010perspective} O. G{\"u}nl{\"u}k and J. Linderoth, ``Perspective reformulations of mixed integer nonlinear programs with indicator variables,'' Mathematical Programming, vol. 124, no. 1, pp. 183--205, 2010.

\bibitem{boyd2004convex} S. Boyd and L. Vandenberghe, ``Convex optimization'', Cambridge University Press, 2004.

% \bibitem{vishnoi2018geodesic} N. K. Vishnoi, ``Geodesic convex optimization: Differentiation on manifolds, geodesics, and convexity'', in \textit{arXiv preprint arXiv:1806.06373}, 2018.

% \bibitem{leboudec2005rate} J. Y. Le Boudec, ``Rate adaptation, congestion control and fairness: A tutorial'', \textit{https://leboudec.github.io/leboudec/latex/cc/LEB3132.pdf}, November 22, 2005.

% \bibitem{johansson1991introduction} P. O. Johansson, ``An Introduction to Modern Welfare Economics'', \textit{Cambridge University Press}, 1991.

% \bibitem{rawls1971theory} J. Rawls, ``A Theory of Justice'', \textit{Harvard University Press}, 1971.

% \bibitem{palomar2006tutorial} D. P. Palomar and M. Chiang, ``A tutorial on decomposition methods for network utility maximization,'' in \textit{IEEE Journal on Selected Areas in Communications}, vol. 24, no. 8, pp. 1439-1451, Jul. 31, 2006.


% \bibitem{gauthier2023architecture} S. Gauthier, G. Vardoyan, and S. Wehner, ``An architecture for control of entanglement generation switches in quantum networks'', in \textit{IEEE Transactions on Quantum Engineering}, vol. 4, no. 01, pp. 1--17, 2023.

% \bibitem{pouryousef2023quantum} S. Pouryousef, H. Shapourian, A. Shabani, and D. Towsley, ``Quantum network planning for utility maximization'', in \textit{Proceedings of the 1st Workshop on Quantum Networks and Distributed Quantum Computing}, pp. 13--18, 2023.

% \bibitem{boyd2007tutorial} S. Boyd, S.-J. Kim, L. Vandenberghe, and A. Hassibi, ``A tutorial on geometric programming'', in \textit{Optimization and Engineering}, vol. 8, pp. 67--127, 2007.


% \bibitem{humphreys2018deterministic} P. C. Humphreys, N. Kalb, J. P. Morits, R. N. Schouten, R. F. Vermeulen, D. J. Twitchen, M. Markham, and R. Hanson, ``Deterministic delivery of remote entanglement on a quantum network,'' in \textit{Nature}, vol. 558, no. 7709, pp. 268–273, 2018.

% \bibitem{shor2000simple} P. W. Shor and J. Preskill, ``Simple proof of security of the BB84 quantum key distribution protocol'', in \textit{Physical Review Letters}, vol. 85, no. 2, pp. 441, 2000.

% \bibitem{nielsenChuang}
% Nielsen, M.A. and Chuang, I., "Quantum computation and quantum information", 2002.

% \bibitem{bennett1996mixed} C. H. Bennett, D. P. DiVincenzo, J. A. Smolin, and W. K. Wootters, ``Mixed-state entanglement and quantum error correction'', in \textit{Physical Review A}, vol. 54, no. 5, pp. 3824, 1996.

% \bibitem{horodecki1999general} M. Horodecki, P. Horodecki, and R. Horodecki, ``General teleportation channel, singlet fraction, and quasidistillation'', in \textit{Physical Review A}, vol. 60, no. 3, pp. 1888, 1999.

% \bibitem{matzner2024topology} 
% R. Matzner, A. Ahuja, R. Sadeghi et. al., 
% ``Topology Bench: systematic graph-based benchmarking for core optical networks,'' 
% \textit{Journal of Optical Communications and Networking}, 
% vol. 17, no. 1, pp. 7--27, 2024.
\bibitem{matzner2024topology} 
R. Matzner, A. Ahuja, R. Sadeghi et al., 
``Topology Bench: systematic graph-based benchmarking for core optical networks,'' 
\textit{Journal of Optical Communications and Networking}, 
vol. 17, no. 1, pp. 7--27, 2024. 
Dataset available at: \url{https://zenodo.org/records/13921775}.

\bibitem{diamond2016cvxpy} S. Diamond and S. Boyd, ``CVXPY: A Python-Embedded Modeling Language for Convex Optimization,'' \textit{Journal of Machine Learning Research}, vol. 17, no. 83, pp. 1--5, 2016.

\bibitem{shor2000simple} P. W. Shor and J. Preskill, ``Simple proof of security of the BB84 quantum key distribution protocol'', in \textit{Physical Review Letters}, vol. 85, no. 2, pp. 441, 2000.


% M. Doherty,  A. Beghelli, S. J. Savory, and P. Bayvel,

% \bibitem{vardoyan2024bipartite} 
% G. Vardoyan, E. van Milligen, S. Guha, S. Wehner, and D. Towsley, 
% ``On the bipartite entanglement capacity of quantum networks'', in 
% \textit{IEEE Transactions on Quantum Engineering}, 
% vol. 5, pp. 1--14, 2024.

% \bibitem{dieudonne1970foundations} J. Dieudonné, ``Foundations of Modern Analysis'', \textit{Read Books Ltd}, 2011

% \bibitem{zbMATH02560682} S. A. Gershgorin, ``{\"Uber die Abgrenzung der Eigenwerte einer Matrix}'', in \textit{Bull. Acad. Sci. URSS}, vol. 1931, no. 6, pp. 749--754, 1931.

\end{thebibliography}
\end{document}